\documentclass[journal]{IEEEtran}

\usepackage{amsmath, amssymb}
\usepackage{amsfonts}
\usepackage{nicefrac}
\usepackage{cite}
\usepackage{epsfig}
\usepackage{subfigure}

\newtheorem{lemma}{Lemma}[section]

\newtheorem{proposition}{Proposition}[section]

\newtheorem{example}{Example}[section]

\title{Stacked OSTBC: Error Performance and Rate Analysis }

\author{Aydin Sezgin,~\IEEEmembership{Member,~IEEE}
        and Oliver Henkel,~\IEEEmembership{Member,~IEEE}% <-this % stops a space
\thanks{Manuscript received March 29, 2006; revised July 10, 2006. This paper was presented in part at the VTC 2004-Fall,
Los Angeles, September 2004.
        The associate editor coordinating the review of this
manuscript and approving it for publication was Dr. R. Michael Buehrer.}% <-this % stops a space
\thanks{A.Sezgin has been with the Fraunhofer-Institute for Telecommunications, Heinrich-Hertz-Institut, Einsteinufer 37,
D-10587 Berlin, Germany. He is now with the Information Systems Laboratory, Stanford University, CA 94305-9510, USA
(e-mail:sezgin@stanford.edu).}
\thanks{O.Henkel is with the Fraunhofer-Institute for Telecomm., HHI, Einsteinufer 37, 10587
Berlin, Germany, (e-mail:henkel@hhi.de).}} 
\begin{document}
\maketitle

\begin{abstract}
It is well known, that the Alamouti scheme is the only space-time code from orthogonal design achieving the capacity of
a multiple-input multiple-output (MIMO) wireless communication system with $n_T=2$ transmit antennas and $n_R=1$
receive antenna. In this work, we propose the $n$-times stacked Alamouti scheme for $n_T=2n$ transmit antennas and show
that this scheme achieves the capacity in the case of $n_R=1$ receive antenna. This result may regarded as an extension
of the Alamouti case. For the more general case of more than one receive antenna, we show that if the number of
transmit antennas is higher than the number of receive antennas we achieve a high portion of the capacity with this
scheme. Further, we show that the MIMO capacity is at most twice the rate achieved with the proposed scheme for all
SNR. We derive lower and upper bounds for the rate achieved with this scheme and compare it with upper and lower bounds
for the capacity. In addition to the capacity analysis based on the assumption of a coherent channel, we analyze the
error rate performance of the stacked OSTBC with the optimal ML detector and with the suboptimal lattice-reduction (LR)
aided zero-forcing detector. We compare the error rate performance of the stacked OSTBC with spatial multiplexing (SM)
and full-diversity achieving schemes. Finally, we illustrate the theoretical results by numerical simulations.
\end{abstract}

\section{Introduction}

Recent information theoretic results have demonstrated that the ability of a system to support a high link quality and
higher data rates in the presence of Rayleigh fading improves significantly with the use of multiple transmit and
receive antennas~\cite{Telatar99,FoschiniGans98}. Since then there has been considerable work on a variety of
schemes~\cite{TirkHotWichBuch} which exploit multiple antennas at both the transmitter and receiver in order to either
obtain transmit and receive diversity and therefore increase the reliability of the system, e.g., orthogonal space-time
block codes (OSTBC) and space-time trellis codes~\cite{TarokhSeCa98,Alamouti,TarokhJafarkCalder99} or achieve the
theoretical bounds~\cite{Foschini96} derived in~\cite{Telatar99,FoschiniGans98}. Interested readers are referred
to~\cite{TirkHotWichBuch}, where a detailed analysis of different schemes is given.

The performance of OSTBC with respect to mutual information has been analyzed (among others)
in~\cite{HassibiHoch2002,Paulraj,Sandhu,Bauch} and it was shown that the capacity is achieved only in the case of
$n_T=2$ transmit, the well known Alamouti scheme~\cite{Alamouti}, and $n_R=1$ receive antennas due to the rate loss
inherent in OSTBC with higher number of transmit antennas. Recently, it was shown in~\cite{SezgJorICASSP05} that due to
this rate loss, OSTBC with odd number of antennas are always outperformed by OSTBC with even number of antennas,
restricting even more the deployment of OSTBC. On the one hand, we have the OSTBC with low complexity and low rates. On
the other hand, we have the space-time trellis codes, which achieve higher spectral efficiency in addition to high
performance with respect to frame error rates. However, the decoding complexity of space-time trellis codes is
increasing exponentially with the number of transmit antennas and the transmission rate. In order to achieve higher
spectral efficiency combined with low complexity maximum likelihood detectors,
\cite{Jafarkhani01,PapadiasFoschi01,TirkkonenHottinenNON,A.Sezgin2004,SezginInfoTheory} designed quasi-orthogonal
space-time block codes (QSTBC) with transmission rate one for more than two transmit antennas.

Other approaches aimed at reducing the decoding complexity of space-time trellis codes. For instance, a layered
space-time architecture was proposed in~\cite{TarokhNagSesCa99}, where the transmit antennas were partitioned into
two-antenna groups and on each group space-time trellis codes were used as component codes. In order to further
decrease the complexity of this layered space-time architecture,~\cite{A.F.Naguib1998,N.Prasad2001,N.Prasad2003} used
the Alamouti scheme as component code for each group in combination with a suboptimal successive group interference
suppression detection strategy. The outage probability of this scheme was analyzed in~\cite{SezginJorsCapHighRate} for
$n_T>n_R$ and an upper bound was derived. An asymptotic analysis of the rate achievable with this scheme is performed
in~\cite{LiYeCapStackedAlamou}. For $n=2$, this transmission scheme is also referred to as double-space-time transmit
diversity (DSTTD) and was proposed as one possible candidate for high speed downlink packet access (HSDPA) in 3GPP and
beyond~\cite{TexasInstrum}.

It is obvious that reducing the computational complexity of the detector without sacrificing much performance is an
important issue. There is a huge amount of suboptimal detectors with low complexity in the literature, linear detectors
like zero-forcing (ZF) or minimum mean square error (MMSE) and nonlinear detectors like e.g.
VBLAST~\cite{WolnianskyFoschiniGoldenValen98}. Unfortunately, these detectors significantly sacrifice performance in
terms of bit-error-rate (BER). Recently, lattice reduction (LR) aided detection in combination with suboptimal
detectors was proposed by Yao and Wornell in order to improve the performance of multi antenna
systems~\cite{YaoWornell}. The lattice reduction algorithm proposed in~\cite{YaoWornell} is optimal, but works only for
MIMO systems with two transmit and two receive antennas. The impact of receive antenna correlation on the performance
of LR-aided detection was analyzed in~\cite{MengPanYouKimLLLforVLBAST}. In~\cite{WindpassingerLLL}, the work
of~\cite{YaoWornell} was extended to systems with more transmit and receive antennas, using the sub-optimal
LLL~\cite{LLL} lattice reduction algorithm. In~\cite{WuebbenLLL}, the LR-aided schemes in~\cite{WindpassingerLLL} were
adopted to the MMSE criterion. Note that the error rate curves of all these LR detectors are parallel to those for
maximum likelihood (ML) detection with only some penalty in power efficiency.

In this work, we show that the stacked Alamouti scheme is capable to achieve the capacity in combination with the
optimal maximum likelihood detector for the case of $n_T=2n$ transmit antennas and $n_R=1$ receive antennas. This was
also shown for the basic Alamouti scheme with $n_T=2$ and $n_R=1$~\cite{HassibiHoch2002}. Our result may therefore be
regarded as an extension of the Alamouti scheme to $n_T>2$. Furthermore, we show that in the case of more than one
receive antenna and if $n_T>n_R$ the stacked Alamouti scheme is capable to achieve a significant portion of the
capacity and approaches the capacity if $n_T\gg n_R$. For any $n_T$, $n_R$, we show that the MIMO capacity is at most
twice the rate achieved with the proposed scheme for all SNR.
However, achieving high portions of the capacity does not guarantee good performance in terms of error probability.
 Thus, we compare the error-rate performance of the proposed scheme with spatial multiplexing (SM), a rate oriented
space-time transmission schemes which achieve a high portion of the capacity of MIMO systems, and with the
aforementioned diversity-oriented QSTBC by deploying LR-aided linear ZF and ML detectors at the receiver, respectively.

The remainder of this paper is organized as follows. In Section~\ref{seq:system_model}, we introduce the system model
and establish the notation. The structure of the stacked Alamouti scheme and the equivalent channel model are shown in
section~\ref{sec:Diversity_analy}. The analysis of the mutual information is presented in section~\ref{sec:MutInfo}.
LR-aided linear ZF detection is shortly described in section ~\ref{sec:Schemes} including the analysis of the
probability density function of the condition number of the equivalent channel generated by the different transmission
schemes (SM,QSTBC, and stacked OSTBC). Section~\ref{seq:simulations} provides simulation results, followed by some
concluding remarks in Section~\ref{sec:Conclusion}.

\section{System model}\label{seq:system_model}
 We consider a system with $n_T$ transmit and $n_R$ receive antennas. Our system model is defined
by
\begin{equation}\label{eq:System}
    \mathbf{Y} = \mathbf{G}_{n_T}\mathbf{H}^T+  \mathbf{N}\;,
\end{equation}
where $\mathbf{G}_{n_T}$ is the ($T \times n_T $) transmit matrix, $\mathbf{Y}=[\mathbf{y}_1,\dots,\mathbf{y}_{n_R}]$
is the ($T \times n_R$) receive matrix, $\mathbf{H}=[\mathbf{h}_1,\dots,\mathbf{h}_{n_T}]$ is a ($n_R \times n_T $)
matrix characterizing the coherent channel, and $\mathbf{N}=[\mathbf{n}_1,\dots,\mathbf{n}_{n_R}]$ is the complex ($T
\times n_R$) white Gaussian noise (AWGN) matrix, where an entry $\{n_{ti}\}$ of $\mathbf{N}$ ($1\leq i \leq n_R$)
denotes the complex noise  at the $i$th receiver for a given time $t\,(1\leq t \leq T)$. The real and imaginary parts
of $n_{ti}$ are independent and $\mathcal{N}$(0,$n_T/(2\mathrm{SNR})$) distributed. An entry of the channel matrix is
denoted by \{$h_{ij}$\}. This represents the complex gain of the channel between the $j$th transmit ($1\leq j \leq
n_T$) and the $i$th receive ($1\leq i \leq n_R$) antenna, where the real and imaginary parts of the channel gains are
independent and normal distributed random variables with $\mathcal{N}$(0,1/2) per dimension. The channel matrix is
assumed to be constant for a block of $T$ symbols and changes independently from block to block.
The average power of
the symbols transmitted from each antenna is normalized to be $\nicefrac{1}{n_T}$, so that the average power of the
received signal at each receive antenna is one and the signal-to-noise ratio (SNR) is $\rho$. It is further assumed
that the transmitter has no channel state information (CSI) and the receiver has perfect CSI.
\section{Code construction}\label{sec:Diversity_analy}
A space time block code is defined by its transmit matrix $\mathbf{G}_{n_T}$ with entries $\{x_j\}_{j=1}^{p}$, which
are elements of the vector $\mathbf{x}=[x_1,\dots,x_p]^T$ with $x_1,\dots,x_{p} \in \mathcal{C}$, where $\mathcal{C}
\subseteq \mathbb{C}$ denotes a complex modulation signal set with unit average power, e.g. $M$-PSK.. The rate $R$ of a
space-time code is defined as $R=p/T$. In this paper, we focus on the rate $n_T/2$ stacked Alamouti scheme. Starting
with the well known (basic) Alamouti scheme~\cite{Alamouti} for $n_T=2$ transmit antennas
\begin{equation}\nonumber
 \mathbf{G}_{2}(x_1,x_2) = \left[
\begin{array}{*{2}{r}}
          x_1 & x_2  \\
        x_2^* & -x_1^*  \\
        \end{array}
\right]\;,
\end{equation}
the transmit matrix of the stacked Alamouti scheme with $n_T=2n$ is constructed in the following way
\begin{align}\nonumber
&\mathbf{G}_{n_T}\left(\{x_j\}_{j=1}^{n_T}\right)  \\
& =\left[\mathbf{G}_{2}(x_1,x_2), \mathbf{G}_{2}(x_3,x_4),\dots,\mathbf{G}_{2}(x_{n_T-1},x_{n_T})\right].\nonumber
\end{align}

\begin{example}
For the case of $n=2$, i.e. $n_T=4$ transmit antennas we have
\begin{equation}\nonumber
\mathbf{G}_{4}(\{x_j\}_{j=1}^{4}) = \left[
\begin{array}{*{4}{r}}
          x_1 & x_2 & x_3 & x_4 \\
        x_2^* & -x_1^* & x_4^* & -x_3^* \\
        \end{array}
\right]\;,
\end{equation}
which is also referred to as DSTTD\cite{TexasInstrum}.
\end{example}

After some manipulations (particularly complex-conjugating) the system model in (\ref{eq:System}) can be rewritten as
\begin{equation}\label{eq:system_H_vorne}
    \mathbf{y}'=\mathbf{H}'\mathbf{x}+ \mathbf{n}'\;,
\end{equation}
where $\mathbf{y}'$, $\mathbf{n}'\in \mathbb{C}^{2n_R}$ and $\mathbf{H}' \in \mathbb{C}^{2n_R\times n_T}$. The
equivalent channel equals
\begin{align}\mathbf{H}'=[(\mathbf{H}_{1}')^T,\dots,(\mathbf{H}_{i}')^T,\dots,(\mathbf{H}_{n_R}')^T]^T, \nonumber
\end{align}
 where $\mathbf{H}_i'$ is given as
\begin{equation}\label{eq:neuer_Kanal}
\mathbf{H}_i'=  \left[ \mathbf{H}_{i,1}',\mathbf{H}_{i,3}',\dots,\mathbf{H}_{i,n_T-1}' \right],
\end{equation}
where
\begin{equation}\nonumber
 \mathbf{H}_{i,j}' = \left[
\begin{array}{*{2}{c}}
          h_{ij} & h_{i(j+1)}  \\
        -h_{i(j+1)}^* & h_{ij}^*  \\
        \end{array}
\right]\;.
\end{equation}

%Oliver
\section{Mutual Information}\label{sec:MutInfo} %--> MutInfoUB, MutInfoLB

The instantaneous capacity $I$ of a MIMO system with $n_T$ transmit and $n_R$ receive antennas is given
as~\cite{Telatar99,FoschiniGans98}
\begin{equation}\label{eq:MutInfoNrGeneral}
I = \log_2\det\left(\mathbf{I}_{n_T} +\frac{\rho}{n_T}\mathbf{H}^H\mathbf{H}\right)\;.
\end{equation}

In the following two subsections, we derive lower and upper bounds for both the ergodic capacity and the average rate
achievable with the proposed stacked scheme in order to yield lower and upper bounds on the ratio of the ergodic
capacity to the average rate of the stacked OSTBC. In the third subsection, we characterize the absolute loss of the
average rate of the stacked OSTBC to the ergodic capacity.

\subsection{Upper bounds on the ergodic capacity and the average rate of stacked OSTBC}\label{sec:MutInfoUB}

By applying the trace-determinant inequality $\det(\mathbf{A})^{1/n}\leq \frac{1}{n}\mathrm{tr}(\mathbf{A})$, we arrive
at a simple upper bound on the instantaneous capacity given as
\begin{equation}\label{eq:UppBoundCapa}
    I \leq I_{ub}= L \log_2
    \bigg(1+\frac{\rho}{n_TL}\underbrace{\sum_{j=1}^{n_T}\sum_{i=1}^{n_R}|h_{ji}|^2}_{\lambda}\bigg)\;,
\end{equation}
where $L$ is equal to $L=\min(n_T,n_R)$. Averaging the upper bound in~\eqref{eq:UppBoundCapa} over all channel
realizations results in~\cite{Gradshteyn} ($C=\mathbb{E}[I]$ denotes ergodic capacity)
\begin{align}\label{eq:UppBoundAnaly}
    C \leq C_{ub} =\mathbb{E}\left[I_{ub}\right] = & \frac{L}{\ln(2)}\sum\limits_{k=0}^{n_Tn_R-1}
    \left(\frac{n_TL}{\rho}\right)^{n_Tn_R-k-1}\\
    & e^{\frac{n_TL}{\rho}} \Gamma\left(1-(n_Tn_R-k),\frac{n_TL}{\rho}\right) \nonumber.
\end{align}
Note that for high SNR, the slope of the upper bound is equal to $L$. In addition to this upper bound, we compare the
rate achieved with the stacked scheme with the following upper bound
\begin{equation}\label{eq:BoundGrant}
    C \leq C_{\mathrm{Jen}} =\log_2\left(\sum_{i=0}^L \binom{L}{i}\frac{K!}{(K-i)!}\left(\frac{\rho}{n_T}\right)^i\right),
\end{equation}
derived in~\cite{GrantUppBound} by using Jensen's inequality, where $K=\max(n_T,n_R)$.

In the following, we analyze the performance of the stacked scheme with respect to mutual information and derive upper
bounds for the average rate of the stacked scheme. We first analyze the case of $n_R=1$ receive antennas and then
generalize the analysis to the case of arbitrary number of receive antennas.

\subsubsection{Case $n_R=1$}

In case of $n_R=1$, the achievable rate of the stacked Alamouti scheme is
\begin{equation}\nonumber
I_{sA} = \frac{1}{2}\log_2\det\left(\mathbf{I}_{n_T} +\frac{\rho}{n_T}(\mathbf{H}_1')^H\mathbf{H}_1'\right)\;.
\end{equation}
Using the determinant equality $\det(\mathbf{I}+\mathbf{AB})=\det(\mathbf{I}+\mathbf{BA})$, after some manipulations we
arrive at
\begin{equation}\label{eq:MutInfoNReq1}
I_{sA} = \log_2\left(1 +\frac{\rho}{n_T} \sum\limits_{j=1}^{n_T}|h_{j1}|^2\right)\;,
\end{equation}
which equals the capacity of a MIMO system with $n_T$ transmit and $n_R=1$ receive antennas~\cite{Telatar99}, i.e. as
long as $n_R=1$, the capacity is achieved for arbitrary $n=n_T/2$. Note that in~\cite[p.199]{TirkHotWichBuch} a Taylor
series expansion is performed for the capacity and the mutual information achievable with certain schemes such as the
stacked OSTBC. After comparing the first two expansion coefficients (the linear term and the second order coefficients)
it is shown that the stacked OSTBC reaches second-order capacity for $n_R=1$, i.e. the second-order coefficient of the
mutual information of the stacked OSTBC is equal to the second-order coefficient of the capacity. Although essential
features of the mutual information can be already seen from the first and second-order coefficients (especially at low
SNR), our result above may regarded as more general, since the exact capacity and mutual information expressions are
analyzed. Further note that the result above may be regarded as an extension of the results in~\cite{HassibiHoch2002}.
There it was shown, that the basic Alamouti scheme with $n_T=2$ and $n_R=1$ achieves the capacity.

\subsubsection{Case of $n_T=4$ and $n_R=2$ (DSTTD)}

In the case of $n_T=4$ transmit and $n_R=2$ receive antennas, the equivalent channel is given by
\begin{equation}\nonumber
   \mathbf{H}'= \left[%
\begin{array}{cccc}
  h_{11} & h_{12} & h_{13} & h_{14} \\
  -h^*_{12} & h^*_{11} &-h^*_{14} & h^*_{13} \\
  h_{21} & h_{22} & h_{23} & h_{24} \\
  -h^*_{22} & h^*_{21} &-h^*_{24} & h^*_{23} \\
\end{array}%
\right].
\end{equation}
The achievable rate in this case is given as
\begin{equation}\nonumber
I_{sA} = \frac{1}{2}\log_2\det\left(\mathbf{I}_{n_T} +\frac{\rho}{n_T}\left[%
\begin{array}{cccc}
  \lambda_1 & 0 & \alpha_1 & \alpha_2 \\
  0 & \lambda_1 & -\alpha_2^* & \alpha_1^* \\
  \alpha_1^* & -\alpha_2 & \lambda_2 & 0 \\
  \alpha_2^* & \alpha_1 & 0 & \lambda_2 \\
\end{array}%
\right]\right),
\end{equation}
where $\lambda_i=\sum_{j=1}^{n_T}|h_{ij}|^2$, $\alpha_1=h_{11}h_{21}^*+h_{12}h_{22}^*+h_{13}h_{23}^*+h_{14}h_{24}^*$,
and $\alpha_2=-h_{11}h_{22}+h_{12}h_{21}-h_{13}h_{24}^*+h_{14}h_{23}$. Using Fischer's inequality
\begin{equation}\nonumber
    \det\left(\left[%
\begin{array}{cc}
  \mathbf{A} & \mathbf{B^H} \\         %Oliver B-->B^H
  \mathbf{B} & \mathbf{D} \\           %Oliver C-->B
\end{array}%
\right]\right)\leq \det(\mathbf{A})\det(\mathbf{D})
\end{equation}
yields
\begin{equation}\nonumber
I_{sA}\leq \log_2\left(\left(1+\frac{\rho}{n_T}\lambda_1\right)\left(1+\frac{\rho}{n_T}\lambda_2\right)\right).
\end{equation}
By using the arithmetic-geometric inequality, we arrive at
\begin{equation}\nonumber
I_{sA}\leq 2\log_2\left(1+\frac{\rho}{2n_T}||\mathbf{H}||^2\right).
\end{equation}
This upper bound equals to twice the rate of a full code rate OSTBC for $n_T=4$ transmit and $n_R=2$ receive antennas
with a power penalty of $3$~dB.
%Oliver
In this particular case a more precise statement can be made due to the following
strict form of Fischer's inequality~\cite{hen.wun-itg05}
\begin{lemma}\label{lem.strictFischer}
   Let
   $\mathbf{P}=\left[
     \begin{array}{cc}
        \mathbf{A} & \mathbf{B^H} \\
        \mathbf{B} & \mathbf{D} \\
     \end{array}%
   \right]$ ($\mathbf{A},\mathbf{D}$ square, nonempty) be positive
   definite. Then
   \begin{equation*}
      \text{$\mathbf{B}$ has full rank}
      \quad\Rightarrow\quad
     \det\mathbf{P} < (\det \mathbf{A}) (\det \mathbf{D})
   \end{equation*}
\end{lemma}
\begin{proof}
   Let $\mathbf{R}\succ0$ denote positive definiteness, and
   $\mathbf{R}\succ \mathbf{S}$ defined by $(\mathbf{R}-\mathbf{S})\succ 0$.
   Then \cite[7.7.6]{HornJohnson}
   $\mathbf{P}\succ0 \Leftrightarrow
    (\mathbf{A}\succ0, \mathbf{D}\succ \mathbf{BA}^{-1}\mathbf{B}^H)$.
   Thus for arbitrary $\mathbf{B}$ holds
   $\mathbf{D}-(\mathbf{D}-\mathbf{BA}^{-1}\mathbf{B}^H)=\mathbf{BA}^{-1}\mathbf{B}^H\succeq0$
   and becomes strict if $\mathbf{B}$ has full rank. Since
   $\big(0\prec \mathbf{S}\prec \mathbf{R} \Rightarrow \det \mathbf{S}< \det \mathbf{R}\big)$
   we obtain
   $\det \mathbf{P}=(\det\mathbf{A})(\det[\mathbf{D}-\mathbf{BA}^{-1}\mathbf{B}^H])
   < (\det \mathbf{A})(\det \mathbf{D})$, if $\mathbf{B}$ has full rank.
\end{proof}
From $\det\mathbf{B}=|\alpha_1|^2+|\alpha_2|^2$ it follows, that
apart from the set of events $\{\alpha_1=\alpha_2=0\}$ of measure zero,
$\mathbf{B}$ has full rank, thus the upper bound for $I_{sA}$ is strict with probability
one.

\subsubsection{Case of arbitrary $n_R$} The available portion of the mutual information achievable with $n_R\geq 1$ for
the stacked Alamouti scheme is
\begin{equation}\label{eq:exactMutInfoStackAlam}
I_{sA} = \frac{1}{2}\log_2\det\left(\mathbf{I}_{n_T} +\frac{\rho}{n_T}(\mathbf{H}')^H\mathbf{H}'\right)\;.
\end{equation}
Following the derivation above for arbitrary $n_R$ results in
\begin{equation}\label{eq:UbbBoundStAl}
    I_{sA}\leq I_{sA}^{ub}=\frac{L_1}{2}\log_2 \left(1+\frac{2\rho}{n_TL_1}||\mathbf{H}||^2\right),
\end{equation}
where $L_1=\min(n_T,2n_R)$. By averaging~(\ref{eq:UbbBoundStAl}) over all channel realizations, an upper bound on the
average rate $R_{sA}^{ub}\geq \mathbb{E}[I_{sA}]$ of the stacked Alamouti scheme similar to~\eqref{eq:UppBoundAnaly}
may be obtained
\begin{align}\label{eq:UppBoundExactNRBel}
    R_{sA} \leq R_{sA}^{ub} =\mathbb{E}\left[I_{ub}\right] = & \frac{L_1}{2\ln(2)}\sum\limits_{k=0}^{n_Tn_R-1}
    \left(\frac{n_TL_1}{2\rho}\right)^{n_Tn_R-k-1}\\
    & e^{\frac{n_TL_1}{2\rho}} \Gamma\left(1-(n_Tn_R-k),\frac{n_TL_1}{2\rho}\right)\nonumber,
\end{align}
which can be approximated using $\log_2(1+x)\approx \log_2(x)$ for $x\gg 1$ by
\begin{equation}\nonumber
    R_{sA}^{ub}\approx
    \frac{L_1}{2}\log_2\left(\frac{2\rho}{n_TL_1}\right)+\frac{L_1}{2\ln(2)}\left(\sum\limits_{p=1}^{n_Tn_R-1}\frac{1}{p}-\gamma\right).
\end{equation}
Note that the approximation gets better for higher SNR and may be inaccurate for low SNR.  Further note that, for high
SNR, the slope of the upper bound~\eqref{eq:UppBoundExactNRBel} and its approximation is equal to $\nicefrac{L_1}{2}$.

\subsection{Lower bounds on the ergodic capacity and the average rate of stacked OSTBC}\label{sec:MutInfoLB}

Similarly to the last subsection, here we derive lower bounds for the ergodic capacity and the average rate of the
stacked OSTBC. Due to the peculiar property of stacked OSTBC, lower bounds are obtained in the procedure for the
following cases: (i) $n_T\leq n_R$, (ii) $n_R < n_T < 2n_R$, (iii) $2n_R\leq n_T \leq 4n_R$,  and (iv) $4n_R < n_T$.

First of all, from~\cite{OymanNBPaulraj} we obtain the following lower bound on the ergodic capacity
\begin{equation}\label{eq:BoundOyman}
    C \geq C_{lb}= \sum_{j=1}^{L}\log_2\left(1+\frac{\rho}{n_T}\exp\left(\sum_{p=1}^{K-j}\frac{1}{p}-\gamma\right)\right),
\end{equation}
where $\gamma\approx 0.57721566$ is Euler's constant.

In order to derive an upper bound on the ratio of the ergodic capacity to the average rate achieved with the stacked
scheme, we need a lower bound for the average rate of the stacked scheme. To this end, we
rewrite~\eqref{eq:exactMutInfoStackAlam} as follows
\begin{align}\label{eq:MutInfoStAlReW}
I_{sA} = \frac{1}{2}\log_2\det\left(\mathbf{I}_{n_T}
+\frac{\rho}{n_T}(\mathbf{H})^H\mathbf{H}+\frac{\rho}{n_T}(\mathbf{H}'_e)^H\mathbf{H}'_e\right),
\end{align}
where $\mathbf{H}$ is the actual MIMO channel, which is obtained by taking the odd rows of the equivalent channel
$\mathbf{H}'$ and $\mathbf{H}_e$ is obtained by taking the even rows of $\mathbf{H}'$. The relation between the actual
channel $\mathbf{H}$ and $\mathbf{H}_e$ is described in the following proposition.
\begin{proposition}\label{prop:RelHandHe} Let $\mathbf{H}_e$ be the even and $\mathbf{H}$ the odd rows of $\mathbf{H}'$ given
in~\eqref{eq:system_H_vorne}, respectively. Then the following holds
\begin{enumerate}
    \item $\nonumber
    \mathbf{H}_e=\mathbf{H}^*\mathbf{J}
$, %where
where\footnote{Notation: $\mathbf{A}^T$, $\mathbf{A}^H$, $\mathbf{A}^*$ means transpose, hermitian
  transpose, and complex conjugation, respectively}
\begin{equation}\nonumber
\mathbf{J}=\mathbf{I}_{\frac{n_T}{2}}\otimes\left[%
\begin{array}{cc}
  0 & 1 \\
  -1 & 0 \\
\end{array}%
\right].
\end{equation}
    \item     $
    \mathbb{E}\left[\mathbf{H}\mathbf{H}_e^H\right]=\mathbb{E}\left[\mathbf{H}\mathbf{J}^T\mathbf{H}^T\right]=\mathbf{0}
    $.
\end{enumerate}
\end{proposition}
\begin{proof}
The proof is straightforward and uninformative and thus it is omitted.
\end{proof}
 Eq.~\eqref{eq:MutInfoStAlReW} can be rewritten as
\begin{align}\nonumber
I_{sA} & =\frac{1}{2}\log_2\Bigg(\det\left(\mathbf{I}_{n_T}
+\frac{\rho}{n_T}(\mathbf{H})^H\mathbf{H}\right) \times\\
&\det\left(\mathbf{I}_{n_T}+\frac{\rho}{n_T}\mathbf{H}'_e\left(\mathbf{I}_{n_T}
+\frac{\rho}{n_T}(\mathbf{H})^H\mathbf{H}\right)^{-1}(\mathbf{H}'_e)^H\right) \Bigg) \nonumber\\
&= \frac{1}{2}\log_2\det\left(\mathbf{I}_{n_T} +\frac{\rho}{n_T}(\mathbf{H})^H\mathbf{H}\right) +\frac{1}{2}\;\times \nonumber \\
 &  \log_2\det\Big(\mathbf{I}_{n_T} + \frac{\rho}{n_T}\mathbf{H}'_e\left(\mathbf{I}_{n_T}
+\frac{\rho}{n_T}(\mathbf{H})^H\mathbf{H}\right)^{-1}(\mathbf{H}'_e)^H\Big)\;.\nonumber
\end{align}
Since $\mathbf{H}'_e\left(\mathbf{I}_{n_T} +\frac{\rho}{n_T}(\mathbf{H})^H\mathbf{H}\right)^{-1}(\mathbf{H}'_e)^H$ is a
positive semidefinite matrix, it follows immediately that the rate achieved with the stacked Alamouti is lower bounded
by
\begin{align}\nonumber
I_{sA} \geq \frac{1}{2}\log_2\det\left(\mathbf{I}_{n_T} +\frac{\rho}{n_T}(\mathbf{H})^H\mathbf{H}\right),
\end{align}
which is half the capacity of a MIMO system with $n_T$ transmit and $n_R$ receive antennas.

Another lower bound is obtained for the case $n_T\leq n_R $ by applying Minkowski's determinant
inequality~\cite[p.482]{HornJohnson}
 ($\det(\mathbf{A}+\mathbf{B}) \geq
  (\det(\mathbf{A})^{\frac{1}{n}}+\det(\mathbf{B})^{\frac{1}{n}})^{n}$,
  $\mathbf{A}\succ0, \mathbf{B}\succeq0$)
to~\eqref{eq:exactMutInfoStackAlam}
\begin{align}
R_{sA}& =\mathbb{E}\left[\frac{1}{2}\log_2\det\left(\mathbf{I}+\frac{\rho}{n_T}(\mathbf{H}')^H\mathbf{H}'\right)\right]
\nonumber \\ & \geq
\frac{n_T}{2}\mathbb{E}\left[\log_2\left(1+\rho\det\left(\frac{1}{n_T}(\mathbf{H}')^H\mathbf{H}'\right)^{\frac{1}{n_T}}\right)\right]\nonumber
\\
 & =\frac{n_T}{2}\mathbb{E}\left[\log_2\left(1+\rho\det\left(\frac{1}{n_T}(\mathbf{H}^H \mathbf{H}+
 \mathbf{H}_e^H\mathbf{H}_e)\right)^{\frac{1}{n_T}}\right)
 \right].\nonumber
 \end{align}
Applying again Minkowski's determinant inequality results in
 \begin{align}
 R_{sA} \geq
\frac{n_T}{2}\mathbb{E}\Bigg[\log_2\Big(1+\rho\det\left(\frac{1}{n_T}\mathbf{H}^H\mathbf{H}\right) ^{1/n_T} \nonumber \\
+\rho\det\left(\frac{1}{n_T}\mathbf{H}_e^H\mathbf{H}_e\right)^{\frac{1}{n_T}}\Big)\Bigg].\nonumber
\end{align}
Since $\mathbf{H}_e$ is obtained simply by conjugating and exchanging some elements of the actual matrix $\mathbf{H}$,
it can be shown that the eigenvalues of $(\mathbf{H}_e)^H(\mathbf{H}_e)$ are the same as the eigenvalues of
$\mathbf{H}^H(\mathbf{H})$. Therefore,  the lower bound is equal to
\begin{align}
 R_{sA}\geq
\frac{n_T}{2}\mathbb{E}\left[\log_2\left(1+\rho\exp\ln\left(2\det\left(\frac{1}{n_T}\mathbf{H}^H\mathbf{H}\right)^{\frac{1}{n_T}}
\right)\right)\right].\nonumber
\end{align}
Since $\log_2(1+ce^x)$ is a convex function in $x$ for $c>0$ and by applying Jensen's inequality it holds that
$\mathrm{E}\left[\log_2(1+ce^x)\right]\geq \log_2(1+c\exp(\mathrm{E}\left[x\right]))$, we have
\begin{align}
R_{sA}& \geq
\frac{n_T}{2}\log_2\left(1+\rho\exp\mathbb{E}\left[\ln\left(2\det\left(\frac{1}{n_T}\mathbf{H}^H\mathbf{H}\right)^{\frac{1}{n_T}}
\right)\right]\right) \nonumber\\
& = \frac{n_T}{2}\log_2\left(1+\rho
2\exp\frac{1}{n_T}\mathbb{E}\left[\ln\left(\det\left(\frac{1}{n_T}\mathbf{H}^H\mathbf{H}\right)
\right)\right]\right).\nonumber
 \end{align}

 From~\cite{OymanNBPaulraj,Goodman}, we know that
\begin{equation}\nonumber
\mathbb{E}\left[\ln\left(\det\left(\frac{1}{n_T}\mathbf{H}^H\mathbf{H}\right)
\right)\right]=\sum_{j=1}^{n_T}\mathbb{E}\left[\ln X_j\right]-n_T\ln n_T,
\end{equation}
where the $X_j$ are independent, $\chi^2$ distributed independent variables with $ 2(n_R-j+1) $ degrees of freedom.
Using this yields
\begin{align}\nonumber
R_{sA} & \geq \frac{n_T}{2}\log_2\left(1+\frac{\rho}{\frac{n_T}{2}}
\exp\left(\frac{1}{n_T}\sum_{j=1}^{n_T}\mathbb{E}\left[\ln X_j\right]\right)\right).
 \end{align}
With
\begin{equation}\nonumber
\mathbb{E}\left[\ln X_j\right]=\psi(n_R-j+1),
\end{equation}
where $\psi(\cdot)$ is the digamma function, which may be rewritten for integer arguments as follows
\begin{equation}\nonumber
\psi(x)=-\gamma+\sum_{p=1}^{x-1}\frac{1}{p}.
\end{equation}
Using this results in the following lower bound for the average rate of the stacked scheme.
\begin{align}\nonumber
    R_{sA}\geq & \frac{n_T}{2}\log_2\left(1+\frac{\rho}{\frac{n_T}{2}}\exp
    \left(\frac{1}{n_T}\sum_{j=1}^{n_T}\sum_{p=1}^{n_R-j}\frac{1}{p}-\gamma\right)\right)\\
& \quad[\text{case }n_T \leq n_R]. \nonumber
\end{align}

Similar steps can be pursued for $n_T\geq 4 n_R$ resulting in the following lower bound
\begin{align}\nonumber
    R_{sA}\geq &
    n_R\log_2\left(1+\frac{2\rho}{n_T}\exp\left(\frac{1}{2n_R}\sum_{j=1}^{2n_R}
    \sum_{p=1}^{\nicefrac{n_T}{2}-j}\frac{1}{p}-\gamma\right)\right)\\
& \quad[\text{case }n_T > 4n_R] \nonumber
\end{align}

For the case of $n_T\geq 2n_R$ we rewrite \eqref{eq:exactMutInfoStackAlam} as
\begin{align}\label{eq:BlockStrukSOSTBCRate}
I_{sA} =\frac{1}{2}\log_2\det\left(\mathbf{I}_{2n_R} +\frac{\rho}{n_T}\left[%
\begin{array}{cc}
  \mathbf{H}\mathbf{H}^H & \mathbf{H}\mathbf{H}_e^H \\
  \mathbf{H}_e\mathbf{H}^H & \mathbf{H}_e\mathbf{H}_e^H \\
\end{array}%
\right]\right).
\end{align}
Since $\mathbb{E}\left[\mathbf{H}\mathbf{H}_e^H\right]=\mathbf{0}$ from proposition~\ref{prop:RelHandHe}, we may
proceed as in~\cite{FoschiniGans98} to arrive at a lower bound given as
\begin{align}\nonumber
I_{sA}\geq \frac{1}{2}\sum_{k=1}^{L_1}\log_2\left(1+\frac{\rho}{n_T}X_{k}\right),
\end{align}
where $X_{k}$ are again independent, $\chi^2$ distributed independent variables with $ 2(K_1-k+1) $ degrees of freedom
with $K_1=\max(2n_R,n_T)$. By following the same line of arguments as in~\cite{OymanNBPaulraj}, we arrive at
\begin{align}\nonumber
    R_{sA} \geq R_{sA}^{lb}=  & \frac{1}{2}\sum_{j=1}^{L_1}\log_2\left(1+\frac{\rho}{n_T}
    \exp\left(\sum_{p=1}^{K_1-j}\frac{1}{p}-\gamma\right)\right)\\
& \quad[\text{case }n_T \ge 2n_R] \nonumber
\end{align}
In~\cite{LiYeCapStackedAlamou}, a similar (however, looser) lower bound was derived for this case in order to analyze
the asymptotic performance (with respect to $\rho$) of stacked OSTBC.

For the case of $n_R<n_T<2n_R$ we have
\begin{align}
R_{sA}&
=\mathbb{E}\left[\frac{1}{2}\log_2\det\left(\mathbf{I}+\frac{\rho}{n_T}(\mathbf{H}')^H\mathbf{H}'\right)\right]\nonumber
\\ & =\mathbb{E}\left[\frac{1}{2}\log_2\det\left(\mathbf{I}+\frac{\rho}{n_T}\left(\mathbf{H} \mathbf{H}^H+
 \mathbf{H}_e\mathbf{H}_e^H\right)\right)\right]\nonumber \\
& =
 \mathbb{E}\left[\frac{1}{2}\log_2\det\left(\frac{1}{2}\mathbf{I}+\frac{\rho}{n_T}\mathbf{H} \mathbf{H}^H+
 \frac{1}{2}\mathbf{I}+\frac{\rho}{n_T}\mathbf{H}_e\mathbf{H}_e^H\right)\right]. \nonumber
\end{align}
Applying now Minkowski's determinant inequality results in
\begin{align}
R_{sA}\geq \frac{1}{2}\mathbb{E}\left[\log_2\det\left(\mathbf{I}+\frac{2\rho}{n_T}\mathbf{H}\mathbf{H}^H\right)\right]
\end{align}
and finally
\begin{align}\nonumber
    R_{sA} \geq R_{sA}^{lb}= & \frac{1}{2}\sum_{j=1}^{L}\log_2\left(1+\frac{2\rho}{n_T}
    \exp\left(\sum_{p=1}^{K-j}\frac{1}{p}-\gamma\right)\right) \\
& \quad[\text{case } n_R< n_T < 2n_R]. \nonumber
\end{align}

The lower bound results derived in this subsection are summarized in Table~\ref{tab:LowBounds} on the top of the next
page.
\begin{table*}
\begin{center}
\begin{tabular}{|c|c|}
  \hline
  Case &  Lower bound on $R_{sA}$\\
  \hline
  $n_T \leq n_R$ &  $\frac{n_T}{2}\log_2\left(1+\frac{\rho}{\frac{n_T}{2}}\exp
    \left(\frac{1}{n_T}\sum_{j=1}^{n_T}\sum_{p=1}^{n_R-j}\frac{1}{p}-\gamma\right)\right)$\\
  $n_R < n_T < 2n_R$ & $\frac{1}{2}\sum_{j=1}^{L}\log_2\left(1+\frac{2\rho}{n_T}
    \exp\left(\sum_{p=1}^{K-j}\frac{1}{p}-\gamma\right)\right)$ \\
  $2n_R \leq n_T $ &  $\frac{1}{2}\sum_{j=1}^{L_1}\log_2\left(1+\frac{\rho}{n_T}
    \exp\left(\sum_{p=1}^{K_1-j}\frac{1}{p}-\gamma\right)\right)$\\
  $4n_R < n_T$ &  $n_R\log_2\left(1+\frac{2\rho}{n_T}\exp\left(\frac{1}{2n_R}\sum_{j=1}^{2n_R}
    \sum_{p=1}^{\nicefrac{n_T}{2}-j}\frac{1}{p}-\gamma\right)\right)$\\
  \hline
\end{tabular}
\caption{Lower bound on $R_{sA}$ for the different cases} \label{tab:LowBounds}
\end{center}
\end{table*}

Note that for high SNR, most of the bounds have a slope equal to $\nicefrac{L_1}{2}$, which equals the slope of the
upper bound~\eqref{eq:UppBoundExactNRBel}. Only for the case $n_R < n_T < 2n_R$, the slope of the lower bound is equal
to $\nicefrac{L}{2}$. In Fig.~\ref{fig:RsAUpLowBounds} on the top of the next page, the average rate, the upper
bound~\eqref{eq:UppBoundExactNRBel} and the lower bounds from Table~\ref{tab:LowBounds} for $n_T=$ and $n_R=1,\dots,4$
are depicted. From the Fig., we observe that the upper bound in~\eqref{eq:UppBoundExactNRBel} and lower bounds track
the average rate quite well. Only in the aforementioned case $n_R < n_T < 2n_R$, the slope of the lower bound differs
from the exact performance and the upper bound. Note that for $n_R=1$, the upper bound coincides with the exact
performance.

\begin{figure*}[htb]
\begin{center}
\subfigure[$n_R=1$]{\includegraphics[scale=0.4]{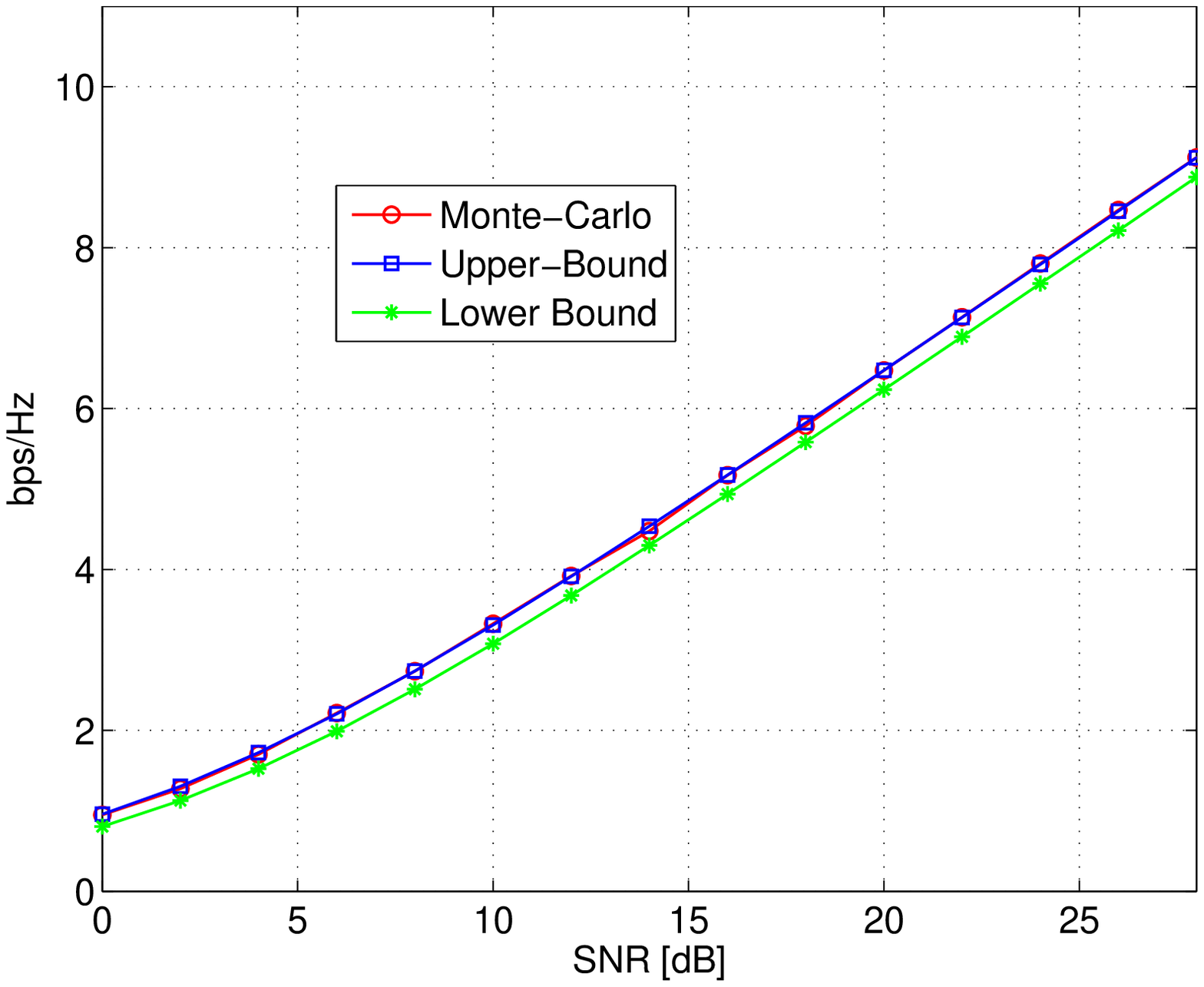}}
\subfigure[$n_R=2$]{\includegraphics[scale=0.4]{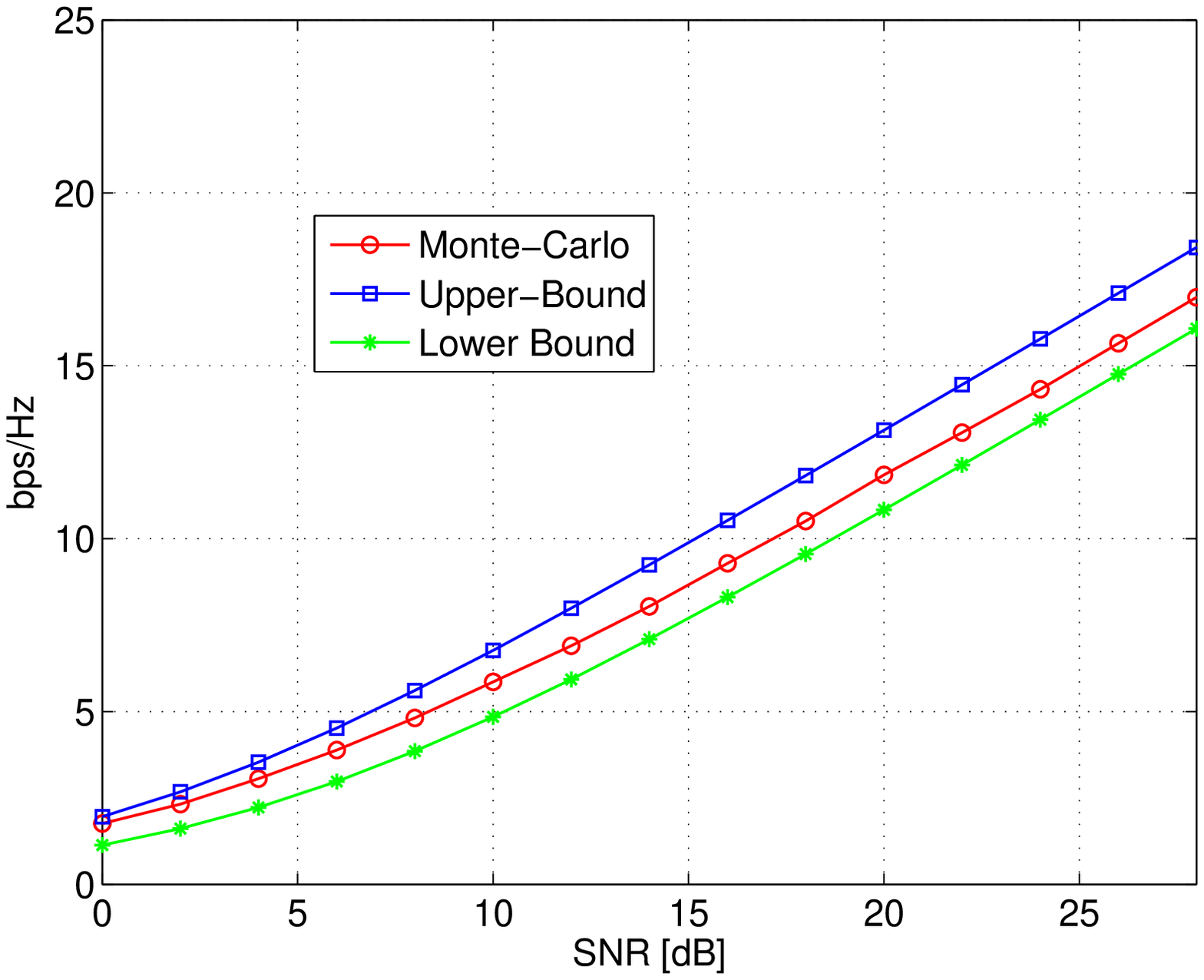}}
\subfigure[$n_R=3$]{\includegraphics[scale=0.4]{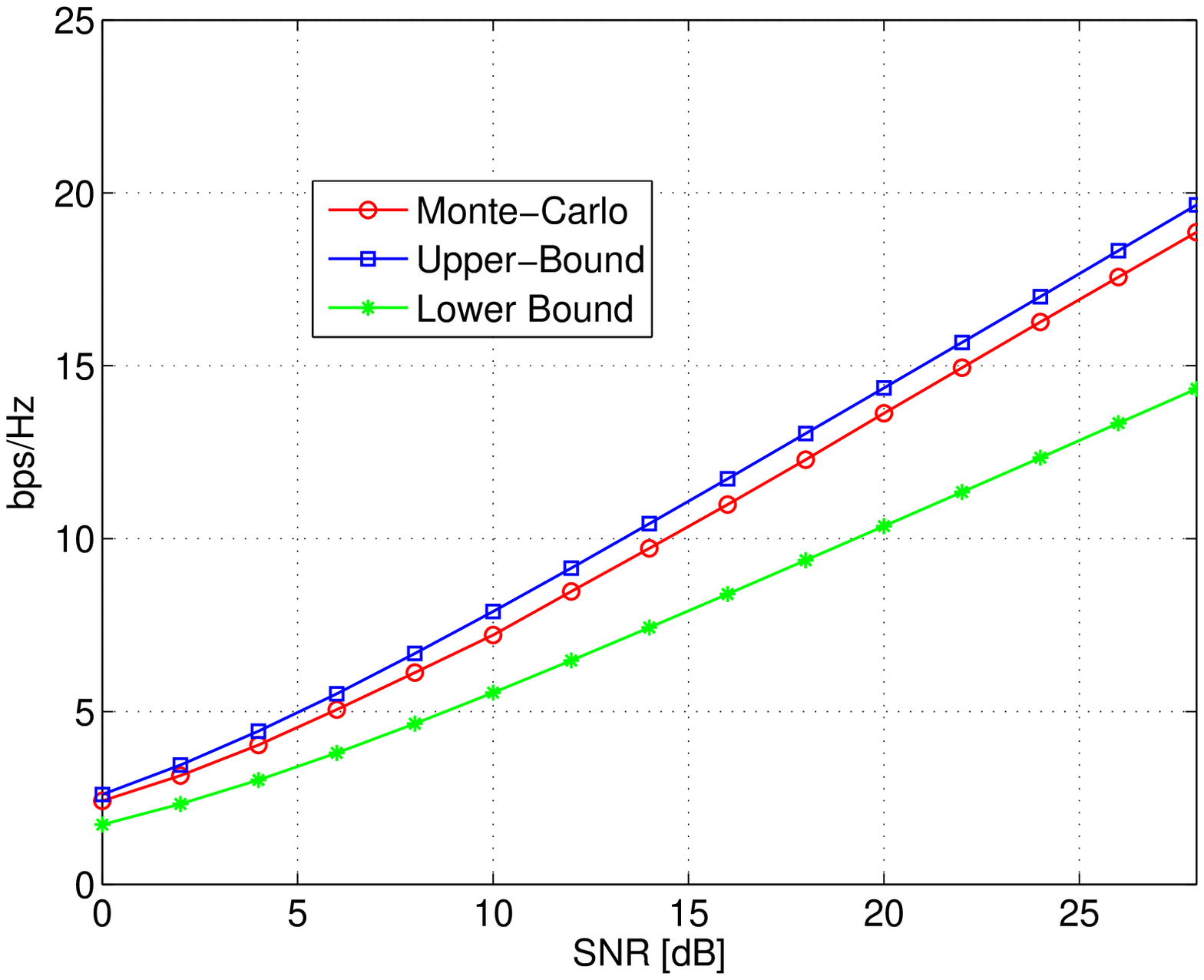}}
\subfigure[$n_R=4$]{\includegraphics[scale=0.4]{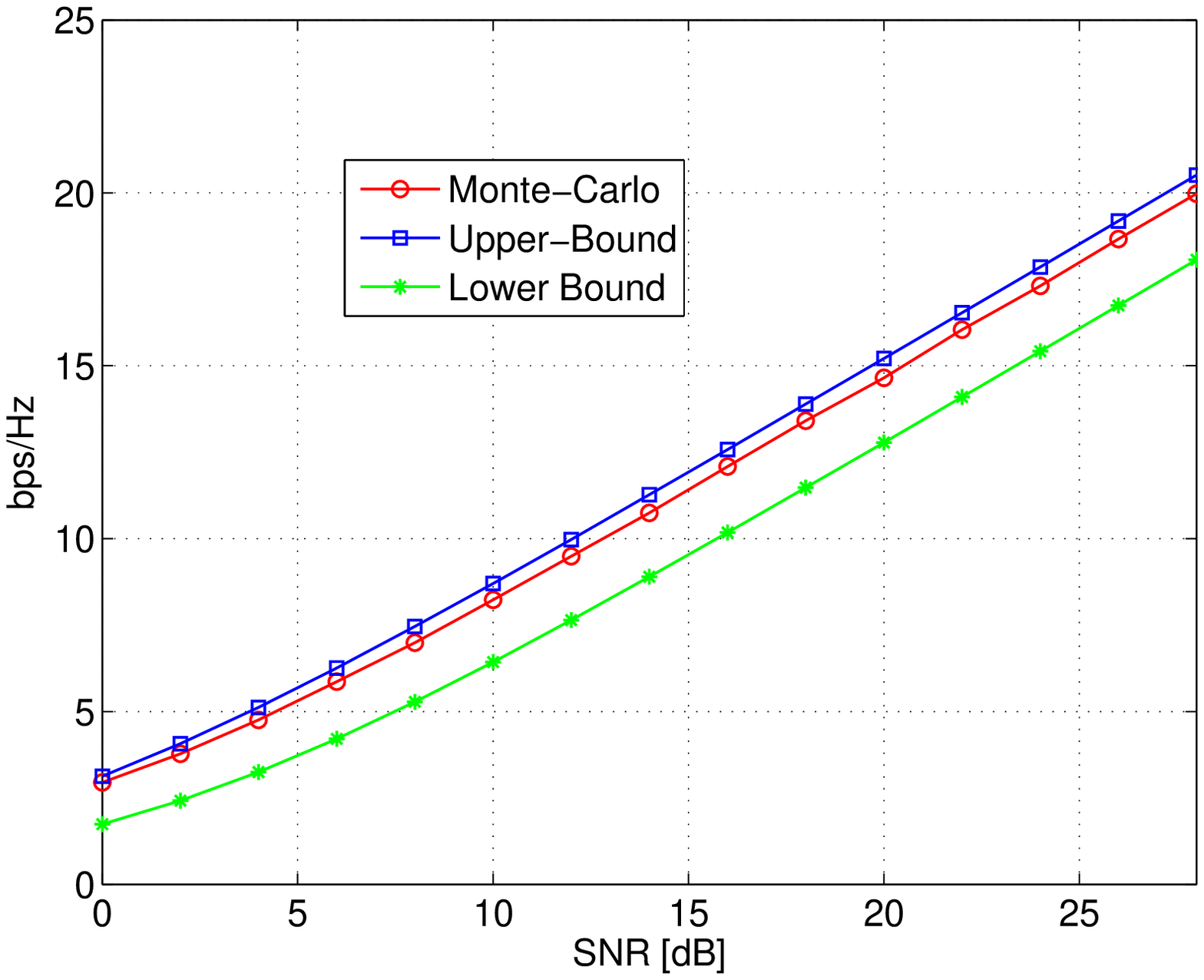}}
\caption{Average rates, upper bounds, and lower bounds of the stacked OSTBCs for $n_T=4$.} \label{fig:RsAUpLowBounds}
\end{center}
\hrule
\end{figure*}

\subsection{Characterization of the absolute rate loss} In this subsection, we characterize the absolute rate loss of the
stacked OSTBC to the ergodic capacity using Fischer's inequality. First of all, we discuss the case of $n_T\geq 2n_R$.
Note that the rate loss with the basic Alamouti scheme ($n_T=2$) was also analyzed in~\cite{Sandhu,HassibiHoch2002}
using different approaches. Starting from~\eqref{eq:BlockStrukSOSTBCRate}, applying Fischer's inequality and averaging
over all channel realizations we arrive at
\begin{align}
R_{sA}  \leq \mathbb{E}\Bigg[ & \frac{1}{2}\log_2\Bigg(\det\left(\mathbf{I}_{n_R}
+\frac{\rho}{n_T}\mathbf{H}\mathbf{H}^H\right) \nonumber \\
& \times \left(\mathbf{I}_{n_R} +\frac{\rho}{n_T}\mathbf{H}_e\mathbf{H}_e^H\right)\Bigg)\Bigg] = C(\rho, n_T,n_R) \label{eq:AbsRatLossNTNR} \\
& \qquad[\text{case } n_T \geq 2n_R] \nonumber,
\end{align}
i.e. as long as $n_T\geq 2n_R$, the average rate of the stacked OSTBC is only upper bounded by the ergodic capacity.
Proceeding similarly for the case $n_T < 2n_R$ results in
\begin{align}
R_{sA}  & =\frac{1}{2}\mathbb{E}\left[\log_2\det\left(\mathbf{I}_{n_T} +\frac{\rho}{n_T}\left[%
\begin{array}{cc}
  \mathbf{\tilde{H}}^H\mathbf{\tilde{H}} & \mathbf{\tilde{H}}^H\mathbf{\tilde{H}}_e \\
  \mathbf{\tilde{H}}_e^H\mathbf{\tilde{H}} & \mathbf{\tilde{H}}_e^H\mathbf{\tilde{H}}_e \\
\end{array}%
\right]\right)\right] \nonumber\\
& \leq \mathbb{E}\left[\log_2\left(\det\left(\mathbf{I}_{\frac{n_T}{2}}
+\frac{\rho}{n_T}\mathbf{\tilde{H}}^H\mathbf{\tilde{H}}\right)\right)\right]\nonumber \\
&=C\left(\frac{\rho}{2},\frac{n_T}{2},2n_R\right)<C\left(\rho,n_T,n_R\right) \;\;\;[\text{case } n_T <
2n_R],\label{eq:AbsRatLossNT22NR}
\end{align}
where $\mathbf{\tilde{H}}$ is obtained by taking the odd columns of the equivalent channel $\mathbf{H}'$ and
$\mathbf{\tilde{H}}_e$ is obtained by taking the even columns of $\mathbf{H}'$. From~\eqref{eq:AbsRatLossNT22NR}, we
observe that for $n_T < 2n_R$ the average rate of the stacked OSTBC is upper bounded by the ergodic capacity of a
system with $\frac{n_T}{2}$ transmit and $2n_R$ receive antennas with a power penalty of $3$~dB.

We can characterize the gap in~\eqref{eq:AbsRatLossNTNR} and~\eqref{eq:AbsRatLossNT22NR} due to the application of
Fischer's inequality. For $n_T < 2n_R$, we have then
\begin{align}
\Delta & = C\left(\frac{\rho}{2},\frac{n_T}{2},2n_R\right)-R_{sA}  \nonumber\\
& =\frac{1}{2}
\mathbb{E}\left[\log_2\left(\frac{\det\left(\mathbf{I}_{n_T}
+\frac{\rho}{n_T}\mathbf{W}_D\right)}{\det\left(\mathbf{I}_{n_T} +\frac{\rho}{n_T}\left[
  \begin{array}{cc}
     \mathbf{\tilde{H}}^H\mathbf{\tilde{H}} & \mathbf{\tilde{H}}^H\mathbf{\tilde{H}}_e \\
     \mathbf{\tilde{H}}_e^H\mathbf{\tilde{H}} & \mathbf{\tilde{H}}_e^H\mathbf{\tilde{H}}_e \\
  \end{array}%
    \right]\right)}\right)\right] \nonumber,
\end{align}
where
\begin{align}\nonumber
\mathbf{W}_D=\left[%
\begin{array}{cc}
  \mathbf{\tilde{H}}^H\mathbf{\tilde{H}} & \mathbf{0} \\
  \mathbf{0} & \mathbf{\tilde{H}}_e^H\mathbf{\tilde{H}}_e \\
\end{array}%
\right].
\end{align}
Since the events of $\mathbf{\tilde{H}}_e^H\mathbf{\tilde{H}}$ having not full rank are of measure zero the strict form
of Fischer's inequality stated in Lemma \ref{lem.strictFischer} shows, that the gap in~\eqref{eq:AbsRatLossNTNR}
and~\eqref{eq:AbsRatLossNT22NR} is non zero in general, i.e. $\Delta>0$, thus it is not possible to reach the upper
capacity bounds.

With
\begin{align}\nonumber
\mathbf{W}_{\mathrm{Off}}=
\left[
  \begin{array}{cc}
     \mathbf{\tilde{H}}^H\mathbf{\tilde{H}} & \mathbf{\tilde{H}}^H\mathbf{\tilde{H}}_e \\
     \mathbf{\tilde{H}}_e^H\mathbf{\tilde{H}} & \mathbf{\tilde{H}}_e^H\mathbf{\tilde{H}}_e \\
  \end{array}
\right] - \mathbf{W}_D,
\end{align}
we can rewrite
\begin{align}
& \det\left(\mathbf{I}_{n_T} + \frac{\rho}{n_T} \left[
  \begin{array}{cc}
     \mathbf{\tilde{H}}^H\mathbf{\tilde{H}} & \mathbf{\tilde{H}}^H\mathbf{\tilde{H}}_e \\
     \mathbf{\tilde{H}}_e^H\mathbf{\tilde{H}} & \mathbf{\tilde{H}}_e^H\mathbf{\tilde{H}}_e \\
  \end{array}
\right]\right) \nonumber \\
& = \det\left(\mathbf{I}_{n_T}
+\frac{\rho}{n_T}(\mathbf{W}_{\mathrm{Off}}+\mathbf{W}_{D})\right) \nonumber \\
& =\det\left(\mathbf{I}_{n_T} +\frac{\rho}{n_T}\mathbf{W}_{D}\right)\times \nonumber\\
& \det\left(\mathbf{I}_{n_T} +\left[\mathbf{I}_{n_T}
+\frac{\rho}{n_T}\mathbf{W}_{D}\right]^{-1}\frac{\rho}{n_T}\mathbf{W}_{\mathrm{Off}}\right)\nonumber
\end{align}
to arrive at
\begin{align}\nonumber
\Delta=-\frac{1}{2} \mathbb{E}\Bigg[\log_2\det\Bigg(\mathbf{I}_{n_T} +\underbrace{\left[\mathbf{I}_{n_T}
+\frac{\rho}{n_T}\mathbf{W}_{D}\right]^{-1}\frac{\rho}{n_T}\mathbf{W}_{\mathrm{Off}}}_{\mathbf{A}}\Bigg)\Bigg].
\end{align}
Using
\begin{align}\nonumber
\det(\mathbf{I}_{n_T}+\mathbf{A})=\exp\left(\sum_{k=1}^{L_1}\ln(1+\mu_{k})\right)
\end{align}
yields
\begin{align}\nonumber
\Delta \leq \frac{1}{2\ln(2)}\mathbb{E}\left[\sum\limits_{k=1}^{L_1}\mu_{k}^2\right]
\end{align}
where the inequality follows from Taylor series expansion $x-\frac{1}{2}x^2\leq \ln(1+x)$ around $x=0$ and the fact
that $\mathrm{tr}(\mathbf{A})=0$, since $\mathbf{A}$ has zero block matrices on its diagonal. Its off-diagonal blocks
have the form $\mathbf{B}=\left[\mathbf{I}_{n_T}
  +\frac{\rho}{n_T}\mathbf{\tilde{H}}^H\mathbf{\tilde{H}}\right]^{-1}
  \frac{\rho}{n_T}\mathbf{\tilde{H}}^H\mathbf{\tilde{H}_e}$ and
$\mathbf{B}_e=\left[\mathbf{I}_{n_T}
  +\frac{\rho}{n_T}\mathbf{\tilde{H}_e}^H\mathbf{\tilde{H}_e}\right]^{-1}
  \frac{\rho}{n_T}\mathbf{\tilde{H_e}}^H\mathbf{\tilde{H}}$,
respectively. Note that the matrices in brackets have the same eigenvalues. This implies that each eigenvalue of
$\mathbf{A}$ appears twice, i.e. $\mu_{k}=\mu_{k+\nicefrac{L_1}{2}}$, $1 \leq k \leq \nicefrac{L_1}{2}$. Additionally
applying the inequality \cite[(5.6.Ex.26)]{HornJohnson} $\mathrm{tr}(\mathbf{A}^2)\le ||\mathbf{A}||^2$ we obtain
$\sum\limits_{k=1}^{L_1}\mu_{k}^2
 =2\sum\limits_{k=1}^{L_1/2}\mu_{k}^2
 \le 2 ||\mathbf{B}||_F^2$.
Further we have
\begin{equation*}
||\mathbf{B}||_F^2=\mathrm{tr}\left\{
  \left(\frac{\rho}{n_T}\right)^2 \mathbf{\tilde{H_e}}\mathbf{\tilde{H_e}}^H
  \mathbf{\tilde{H}}
    \left[\mathbf{I}_{n_T}
      +\frac{\rho}{n_T}\mathbf{\tilde{H}}^H\mathbf{\tilde{H}}\right]^{-2}
  \mathbf{\tilde{H}}^H
\right\}
\end{equation*}
which can be interpreted as the trace of a product of two positive semi definite matrices $\mathbf{P}$, $\mathbf{Q}$.
Using the fact, that $\mathbf{\tilde{H_e}}\mathbf{\tilde{H_e}}^H$ has the same ordered eigenvalues as
$\mathbf{\tilde{H}}\mathbf{\tilde{H}}^H$ and the inequality $\mathrm{tr}(\mathbf{PQ})\le\sum_j
\mu_k(\mathbf{P})\mu_k(\mathbf{Q})$ \cite{MarshallOlkin} yields $\sum\limits_{k=1}^{L_1}\mu_{k}^2\le L_1$ and we arrive
at the final bound
\begin{align}\nonumber
\Delta \leq \frac{1}{\ln(2)}\mathbb{E}\left[
  \sum\limits_{k=1}^{\nicefrac{L_1}{2}}\mu_{k}^2\right]
  \le \frac{L_1}{2 \ln(2) }.
\end{align}

In addition to that, we have the loss
between $C\left(\frac{\rho}{2},\frac{n_T}{2},2n_R\right)$ and $C\left(\rho,n_T,n_R\right)$.
Approximating~\eqref{eq:BoundGrant} and~\eqref{eq:BoundOyman} for high SNR as
\begin{align}
& C_{\mathrm{Jen}} \left(\rho,n_T,n_R\right)= \log_2\left(1+ \sum_{i=1}^L
\binom{L}{i}\frac{K!}{(K-i)!}\left(\frac{\rho}{n_T}\right)^i\right) \nonumber\\
 \approx & \log_2\left(1+ \frac{K!}{(K-L)!}
\left(\frac{\rho}{n_T}\right)^L\right) \nonumber\\
 = & \log_2\left( \frac{(K-L)!}{K!}+ \left(\frac{\rho}{n_T}\right)^L\right)+\log_2\left(\frac{K!}{(K-L)!} \right) \nonumber\\
 \approx  & \log_2\left( 1+ \left(\frac{\rho}{n_T}\right)^L\right) \approx   \log_2\left( \left( 1+
\frac{\rho}{n_T}\right)^L\right) \nonumber \\
= & L\log_2\left(1+\frac{\rho}{n_T}\right)\label{eq:ApproxCapNTNR}
\end{align}
and
\begin{align}
& C\left(\frac{\rho}{2},\frac{n_T}{2},2n_R\right) \nonumber \\
 \overset{(a)}{\geq} &
\frac{n_T}{2} \log_2\left(1+\frac{\rho}{n_T}\exp\left(\frac{2}{n_T} \sum_{j=1}^{\frac{n_T}{2}}\sum_{p=1}^{2n_R-j}\frac{1}{p}-\gamma\right)\right) \nonumber\\
 = & \frac{n_T}{2} \log_2\left(\exp\left(-\frac{2}{n_T}
\sum_{j=1}^{\frac{n_T}{2}}\sum_{p=1}^{2n_R-j}\frac{1}{p}+\gamma\right)+\frac{\rho}{n_T}\right) \nonumber \\
& + \frac{n_T}{2\ln(2)} \left(\frac{2}{n_T}
\sum_{j=1}^{\frac{n_T}{2}}\sum_{p=1}^{2n_R-j}\frac{1}{p}-\gamma\right) \nonumber \\
\approx & \left(\frac{n_T}{2}\right)\log_2\left(1+\frac{\rho}{n_T}\right),\label{eq:ApproxCapNT22NR}
\end{align}
where $(a)$ follows from applying Jensen's inequality to~\eqref{eq:BoundOyman}. With~\eqref{eq:ApproxCapNTNR}
and~\eqref{eq:ApproxCapNT22NR}, the loss between $C\left(\frac{\rho}{2},\frac{n_T}{2},2n_R\right)$ and
$C\left(\rho,n_T,n_R\right)$ is quite accurately described by
\begin{align}\nonumber
& C\left(\rho,n_T,n_R\right)-C\left(\frac{\rho}{2},\frac{n_T}{2},2n_R\right) \\
& \approx
 \left(L-\frac{n_T}{2}\right)\log_2\left(1+\frac{\rho}{n_T}\right), \quad[\text{case }n_T < 2n_R] \nonumber.
\end{align}
Finally, the absolute loss for $n_T < 2n_R$ between the ergodic capacity of a MIMO system and the stacked scheme is
given by
\begin{align}\nonumber
& \left(L-\frac{n_T}{2}\right)\log_2\left(1+\frac{\rho}{n_T}\right) \leq C\left(\rho,n_T,n_R\right)-R_{sA} \\
& \leq \frac{n_T}{2 \ln(2) } + \left(L-\frac{n_T}{2}\right)\log_2\left(1+\frac{\rho}{n_T}\right) \nonumber.
\end{align}

The same procedure can be pursued for $n_T\geq 2n_R$ resulting in the following general characterization for any
$n_T,n_R$
\begin{align}
\max\left(0,L-\frac{n_T}{2}\right)\log_2\left(1+\frac{\rho}{n_T}\right) \leq C\left(\rho,n_T,n_R\right) &-R_{sA} \nonumber \\
\leq \frac{L_1}{2 \ln(2) }  + \max\left(0,L-\frac{n_T}{2}\right)\log_2\left(1+\frac{\rho}{n_T}\right),\nonumber
\end{align}
which is equal to
\begin{align}
& \left(L-\frac{L_1}{2}\right)\log_2\left(1+\frac{\rho}{n_T}\right) \leq C\left(\rho,n_T,n_R\right) &-R_{sA} \nonumber
\\
& \leq \frac{L_1}{2 \ln(2) }  +
\left(L-\frac{L_1}{2}\right)\log_2\left(1+\frac{\rho}{n_T}\right).\label{eq:AbsLossAllg}
\end{align}
From~\eqref{eq:AbsLossAllg}, we observe that as long as $n_T \geq 2n_R$, the absolute loss is only a constant, which
depends only on the number of receive antennas. In case $n_T < 2n_R$ the absolute loss increases linearly with
$\left(L-\frac{n_T}{2}\right)$.
\section{Suboptimal detection and condition number}\label{sec:Schemes}
In the previous sections, we have shown that the stacked OSTBC achieves significant portions of the ergodic capacity.
This does not, however, guarantee good performance in terms of error probability, which will be investigated in this
section. Note that in the analysis in the previous sections it was implicitly assumed, that an optimal
maximum-likelihood detector is used at the receiver, which performs an exhaustive search over all possible transmit
symbols at each detection step. Especially for higher number of transmit antennas, this becomes computationally
prohibitive. If additionally high rates are requested, then higher order modulation sizes are necessary which increases
the computational complexity even more. Thus, suboptimal detection schemes have to be employed reducing the detection
complexity and thereby achieving reasonable error rate performance results. Therefore, in this section the impact of
the suboptimal LR-aided linear ZF-detector on the performance of the stacked OSTBC is analyzed and compared to SM and
QSTBC by resorting the equivalent channel representation. In order to apply the LR algorithm, the system model has to
rewritten, which is done in the following subsections for the different transmission schemes. Afterwards, the LR-aided
linear ZF-detection is described briefly.
\subsection{Spatial Multiplexing (SM)} For SM, the transmit matrix $\mathbf{G}_{n_T}$ is reduced to
$\mathbf{x}$, since $T=1$. In order to apply the suboptimal LR for SM, the system model in (\ref{eq:System}) has to be
rewritten as a real model~\cite{WindpassingerLLL} of the form
\begin{equation}\nonumber
 \mathbf{y}_E = \left[%
\begin{array}{cc}
  \Re\{\mathbf{x}\} \\
  \Im\{\mathbf{x}\} \\
\end{array}%
\right]^T\mathbf{H}^{SM}_E+  \mathbf{n}_E\;,
\end{equation}
where
\begin{equation}\nonumber
 \mathbf{y}_E=\left[%
\begin{array}{c}
  \Re\{\mathbf{y}\} \\
  \Im\{\mathbf{y}\} \\
\end{array}%
\right]^T\;,\mathbf{n}_E=\left[%
\begin{array}{c}
  \Re\{\mathbf{n}\} \\
  \Im\{\mathbf{n}\} \\
\end{array}%
\right]^T\;,
\end{equation}
and
\begin{equation}\nonumber
\mathbf{H}^{SM}_E=\left[%
\begin{array}{rr}
  \Re\{\mathbf{H}\} & \Im\{\mathbf{H}\} \\
  -\Im\{\mathbf{H}\} & \Re\{\mathbf{H}\} \\
\end{array}%
\right]\;.
\end{equation}
In the following, we refer to $\mathbf{H}^{SM}_E$ as the equivalent channel for the SM scheme.

\subsection{QSTBC}

Without loss of generality, in this subsection we shortly describe the QSTBC for $n_T=4$ transmit
antennas~\cite{SharmaPapadias02}. To generalization to higher number of transmit antennas is
straightforward~\cite{A.Sezgin2004}.  The transmit matrix for $n_T=4$ transmit antennas is then
given~\cite{SharmaPapadias02,A.Sezgin2004}.
\begin{eqnarray}
\mathbf{G}_{4}({\mathbf{x}}) = \left[
\begin{array}{*{4}{r}}
          x_1 & x_2 & x_3 & x_4 \\
        x_2^* & -x_1^* & x_4^* & -x_3^* \\
        x_3 & -x_4 & -x_1 & x_2 \\
       x_4^* & x_3^* & -x_2^* & -x_1^* \\
        \end{array}\nonumber
\right]\;.
\end{eqnarray}
After rewriting (\ref{eq:System}), we arrive at (similar to the proposed scheme, (cf.~\eqref{eq:system_H_vorne})
\begin{equation}\label{eq:EQ_SysModQSTBC}
    \mathbf{y}^Q=\mathbf{H}^Q\mathbf{x}+ \mathbf{n}^Q\;,
\end{equation}
where $\mathbf{H}^Q=[(\mathbf{H}_1^Q)^T,\dots,(\mathbf{H}_i^Q)^T,\dots,(\mathbf{H}_{n_R}^Q)^T]^T$ and
$(\mathbf{H}_i^Q)$ is given as
\begin{equation}
\mathbf{H}_i^Q=\left[%
\begin{array}{rrrr}
   h_{1i}   &  h_{2i}   &   h_{3i}   &  h_{4i}   \\
  -h_{2i}^* &  h_{1i}^* &  -h_{4i}^* &  h_{3i}^* \\
  -h_{3i}   &  h_{4i}   &   h_{1i}   & -h_{2i}   \\
  -h_{4i}^* & -h_{3i}^* &   h_{2i}^* &  h_{1i}^* \\
\end{array}%
\right] \nonumber \;.
\end{equation}
 For general $n_T$, we have to rewrite the system model in~\eqref{eq:EQ_SysModQSTBC} as a real model similar to SM.
 For $n_T=4$, however, it is not necessary to resort to the real system model. Here, the system model can be decomposed
 such that the iterative optimal algorithm in~\cite{YaoWornell} for a system
with $n_T=2$ transmit antennas can be applied. For this we first perform channel-matched filtering as the first stage
and noise pre-whitening as the second stage of preprocessing at the receiver resulting in two independent
subsystems~\cite{SharmaPapadias03}, one of which
\begin{equation}\nonumber
    \tilde{\mathbf{y}}_o = \underbrace{\left[%
\begin{array}{rr}
  \beta    &     \jmath\beta   \\
  \epsilon & -\jmath\epsilon \\
\end{array}%
\right]}_{\mathbf{H}_E^{Q}} \left[\begin{array}{c}
      x_1 \\
      x_3 \\
    \end{array}\right] + \tilde{\mathbf{n}}_o \;,
\end{equation}
is only a function of the elements of $\mathbf{x}$ with odd index, and the other one is only a function of the elements
of $\mathbf{x}$ with even index,
\begin{equation}\nonumber
    \tilde{\mathbf{y}}_e = \underbrace{\left[%
\begin{array}{rr}
  \beta    &     \jmath\beta   \\
  \epsilon & -\jmath\epsilon \\
\end{array}%
\right]}_{\mathbf{H}_E^{Q}} \left[\begin{array}{c}
      x_4 \\
      x_2 \\
    \end{array}\right] + \tilde{\mathbf{n}}_e \;,
\end{equation}
where $\mathbf{H}_E^{Q}$ is the $2\times 2$ equivalent channel for QSTBC, $\beta = \sqrt{\frac{\lambda+\alpha}{2}}$,
$\epsilon = \sqrt{\frac{\lambda-\alpha}{2}}$, $\lambda = \sum_{i=1}^{n_R}\sum_{j=1}^{n_T} |h_{i,j}|^2$, and $\alpha =
\sum_{i=1}^{n_R}2\mathrm{Im}(h_{i,1}^*h_{i,3}+h_{i,4}^*h_{i,2})$. Both subsystems can now be detected separately, which
reduces the complexity of the receiver significantly.
\begin{lemma}
In order to get the best performance with respect to error rates and a decoupled system with scalar input and scalar
output as in the case of OSTBC, the columns of $\mathbf{H}_E^{Q}$ have to be orthogonal. However, the probability that
this occurs for $\mathbf{H}_E^{Q}$ is zero.
\end{lemma}
\begin{proof}
For orthogonality, it follows from the scalar product of the columns of $\mathbf{H}_E^{Q}$ that $\alpha$ has to be
zero. But since the channel entries $\{hji\}$ are mutually independent and identically distributed (i.i.d.) random
complex Gaussian processes, the probability $P_r(\alpha=0)$ is equal to the probability
$P_r(\sum_{i=1}^{n_R}2\mathrm{Im}(h_{i,1}^*h_{i,3}+h_{i,4}^*h_{i,2})=0)$, which in turn is zero. From this it follows
that orthogonality and therefore a decoupled system can not be achieved.
\end{proof}
A disadvantage of this QSTBC is that in order to achieve the same transmission rate as SM, we have to compensate the
rate loss by using a considerably higher constellation. But recall that higher constellations complicates
amplification, synchronization, and detection. E.g., a transmission rate of 4 bits/sec/Hz for a system with $n_T=4$
transmit antennas is achieved by SM with BPSK, whereas 16QAM is required for the code rate one QSTBC.
In~\cite{PapadiasFoschi01,SezginHenkelAsi05} it was shown that QSTBC approach the capacity in case of $n_R=1$, which is
achieved in case of the stacked OSTBC as shown in section~\ref{sec:MutInfoUB}. For $n_R>1$, the performance of QSTBC in
terms of mutual information degrades severely in contrast to the stacked OSTBC, which achieve at least half of the
capacity as derived in section~\ref{sec:MutInfoLB}.

\subsection{Proposed scheme}
Given~\eqref{eq:system_H_vorne}, the equivalent real signal model for the proposed stacked OSTBC is given as
\begin{equation}\nonumber
    \mathbf{y}'' = \mathbf{H}_E^{OS}\left[%
\begin{array}{c}
  \Re\{\mathbf{x}\} \\
  \Im\{\mathbf{x}\} \\
\end{array}%
\right]+  \mathbf{n}''\;,
\end{equation}
where
\begin{equation}\nonumber
\mathbf{H}_E^{OS}=\left[%
\begin{array}{rr}
  \Re\{\mathbf{H}'\} & -\Im\{\mathbf{H}'\} \\
  \Im\{\mathbf{H}'\} &  \Re\{\mathbf{H}'\} \\
\end{array}%
\right]\;.
\end{equation}

\subsection{LR-aided linear ZF Detection}
By applying the algorithm, the $m\times n$ equivalent channel $\mathbf{H}_E$ for each transmission scheme can be
decomposed as
\begin{equation}\label{eq:H_E_decomp}
\mathbf{H}_E=\mathbf{Q}\mathbf{R}\;,
\end{equation}
where $\mathbf{R}$ is a $n\times n$  matrix with integer entries and $\mathbf{Q}$ is a $m\times n$ matrix, which is
better conditioned than $\mathbf{H}_E$, i.e. the columns of $\mathbf{Q}$ are less correlated and shorter. A good
indication for the correlation of a matrix is the so called condition number, which is defined as the ratio of the
largest singular value of the matrix to the smallest. Using (\ref{eq:H_E_decomp}), the equivalent signal model is then
given as
\begin{equation}\nonumber
    \mathbf{y} = \mathbf{H}_E\mathbf{x}_r+  \mathbf{n}=
    \mathbf{QR}\mathbf{x}_r+  \mathbf{n}=\mathbf{Q}\mathbf{z}+  \mathbf{n}\;.
\end{equation}
Now, by multiplying $\mathbf{Q}^{-1}$ from left to $\mathbf{y}$ we arrive at
\begin{equation}\nonumber
    \mathbf{\tilde{y}} =\mathbf{z}+  \mathbf{Q}^{-1}\mathbf{n}\;,
\end{equation}
where the noise enhancement and coloring is relatively small, since $\mathbf{Q}^{-1}$ is also good conditioned. In
order to get a estimation for the transmitted symbols, the following operation has to be applied
\begin{equation}\label{eq:DetecScalShiftQuant}
    \mathbf{\hat{x}} = C\left(\mathbf{R}^{-1}\mathcal{Q}_{\mathbb{Z}^n}
    \left[\frac{1}{C}\mathbf{\tilde{y}}-\mathbf{R}\frac{1}{2}\mathbf{1}_n\right]+ \frac{1}{2}\mathbf{1}_n\right)\;,
\end{equation}
where $\mathbf{1}_n$ is a $n\times 1$ vector of ones, $C$ is a constant given as $C=\sqrt{\frac{6}{M-1}}$ and
$\mathcal{Q}_{\mathbb{Z}^n}[\cdot]$ describes the component-wise quantization with respect to the infinite integer
space $\mathbb{Z}$. However, this quantization can only be applied, if the transmit modulation signal set $\mathcal{C}$
is transformed to $\mathbb{Z}$, which is achieved with the scaling and shifting of $\mathbf{\tilde{y}}$ within the
quantization operation in~\eqref{eq:DetecScalShiftQuant}. Note that after this quantization, re-scaling and
re-shifting, some points may lie outside the constellation. A suboptimal solution is to assign these points to the
nearest point within the constellation. For BPSK, the effect of this assignment has a significant effect on the error
rate performance, however, this gain diminishes with higher order modulations.

\subsection{Condition number}\label{sec:cond_number} For illustration,
\begin{figure}[htb]
     \centering
  \includegraphics[scale=0.5]{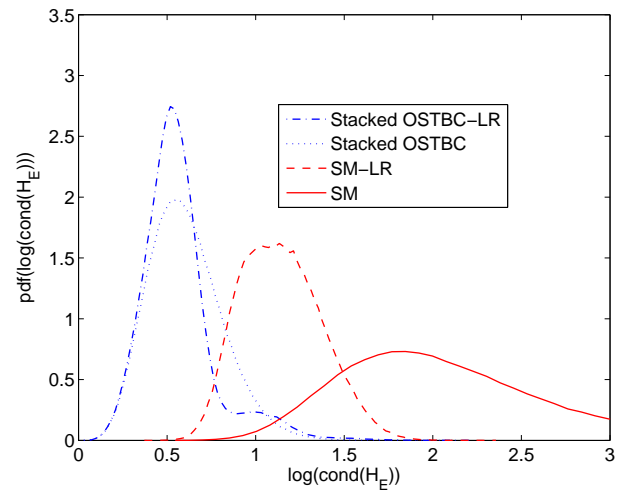}
  \caption{Pdfs of channel cond. numbers with SM or the stacked OSTBC with and w/o LR for a $4\times 4$ system.} \label{fig:Prob_Dens_Cond_NTHalf}
\end{figure}
the probability density functions (pdfs) of the natural logarithm of the condition number of the channels for the
stacked QSTBC and SM are depicted in Fig.~\ref{fig:Prob_Dens_Cond_NTHalf}.  From the Fig., we observe that the
SM-channel is bad-conditioned and that LR has a great impact on the channel. For the stacked OSTBC, we observe that the
impact of LR is not as significant as for SM.

The pdf of the natural logarithm of the condition number for the QSTBC is depicted in Fig.~\ref{fig:Prob_Dens_Cond}.
For comparison, the pdf for the stacked OSTBC is also plotted. In case of QSTBC, for some channels we have no gain with
LR, since many samples of the equivalent channel generated with QSTBC have inherently low condition numbers such that
the LR has no effect. Different from the QSTBC, for the stacked OSTBC there is a gain achieved by applying the LR for
almost all samples of the equivalent channel model. Note that for orthogonal channels (e.g., with OSTBC), the pdf is a
dirac impulse at position $0$.
\begin{figure}[htb]
     \centering
  \includegraphics[scale=0.5]{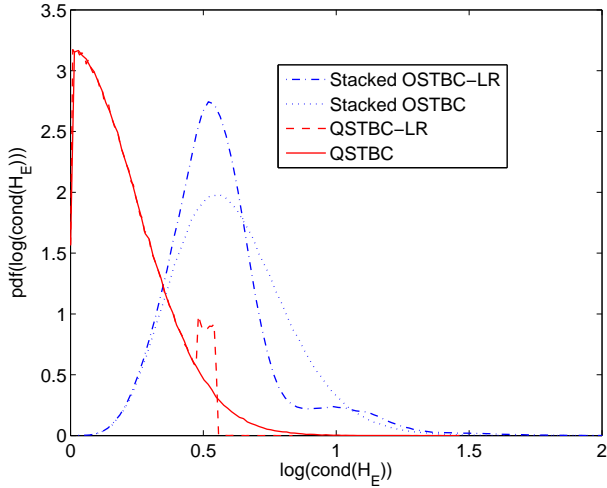}%scale=0.5
  \caption{Pdfs of channel cond. numbers with the stacked OSTBC or  QSTBC with and w/o LR.} \label{fig:Prob_Dens_Cond}
\end{figure}

\section{Simulations}\label{seq:simulations}
In Fig.~\ref{fig:capacity_nr1}, the average rate of the stacked Alamouti scheme and the ergodic capacity of a MIMO
system  with $n_R=2$ and $n_T=2,4$ and $n_T=8$ is depicted. In case of $n_T=2$, we have the standard Alamouti scheme.
From the Fig., we observe that the difference between the average rate of the stacked Alamouti scheme and the capacity
diminishes significantly by increasing the number of transmit antennas.
\begin{figure}[htb]%htbp
     \centering
  \includegraphics[scale=0.5]{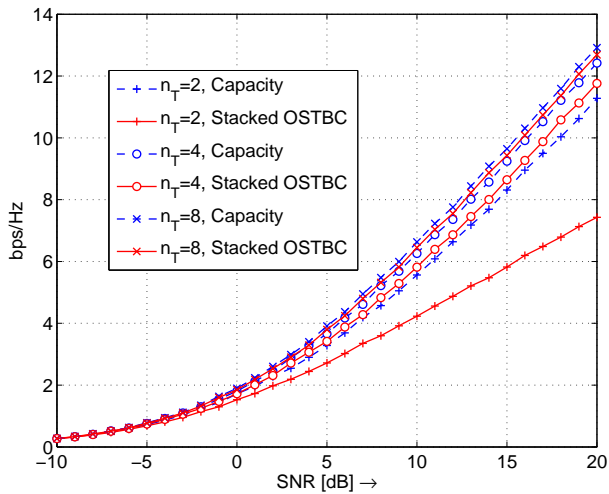}%0.5
  \caption{Ergodic capacity and average rates of the stacked OSTBC
  with $n_R=2$ receive and $n_T=2$,$n_T=4$ and $n_T=8$ transmit antennas.}
   \label{fig:capacity_nr1}
\end{figure}

In Fig.~\ref{fig:capacity_nr2}, the average rate of the stacked Alamouti scheme and the ergodic capacity with $n_T=4$
and $n_R=2,4$ and $n_T=8$ is depicted.
\begin{figure}[htb]%htbp
     \centering
  \includegraphics[scale=0.5]{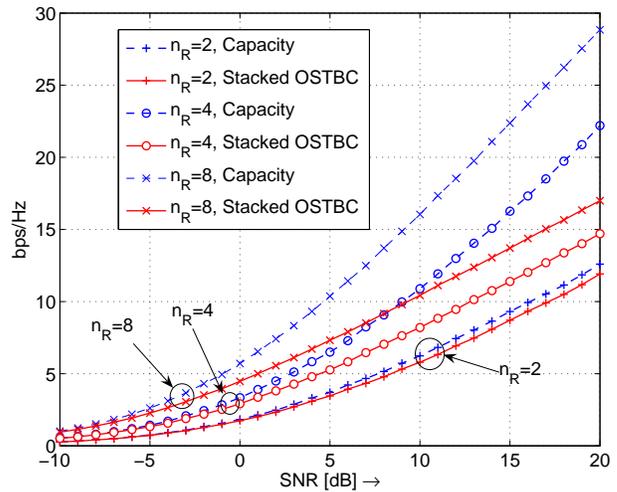}%0.5
  \caption{Ergodic capacity and average rates of the stacked OSTBC
  with $n_T=4$ transmit and $n_R=2$,$n_R=4$ and $n_R=8$ receive antennas.}
   \label{fig:capacity_nr2}
\end{figure}
In contrast to the case of increasing number of transmit antennas, here we observe that the difference between the
average rate of the stacked Alamouti scheme and the ergodic capacity increases by increasing the number of receive
antennas.

In Fig.~\ref{fig:Ratio}, the ratio $C/R_{sA}$ is depicted for $n_T=8$ transmit and $n_R=2$ (bottom) to $n_R=9$ (top)
receive antennas. For high SNR, we observe that as long as $n_T \geq 2n_R$ the ratio decreases as the SNR increases. In
case $n_T< 2n_R$ the ratio increases steadily. As derived in section~\ref{sec:MutInfoLB}, the ratio is upper bounded by
$C/R_{sA}<2$ for any $n_R$, $n_T$.

\begin{figure}[htb]%htbp
     \centering
  \includegraphics[scale=0.5]{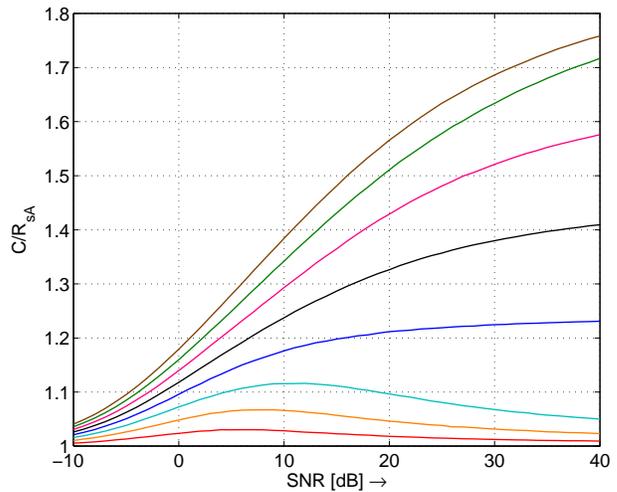}%0.5
  \caption{Ratio $C/R_{sA}$ for $n_T=8$ transmit and $n_R=2$ (bottom) to $n_R=9$ (top) receive antennas.}
   \label{fig:Ratio}
\end{figure}

In Fig.~\ref{fig:RatioAnaly}, the ratio $C/R_{sA}$ is depicted for $n_T=8$ transmit and $n_R=4$, $n_R=6$ and $n_R=9$
receive antennas. In addition to that, we used our lower and upper bounds derived in the previous section in order to
derive lower and upper bounds for the ratio $C/R_{sA}$, i.e.
\begin{align}
\frac{C_{lb}}{R_{sA}^{ub}}\leq \frac{C}{R_{sA}}\leq \frac{C_{\mathrm{Jen}}}{R_{sA}^{lb}}
\end{align}
 Based on the derivations in section~\ref{sec:MutInfoLB}, we know that the
ratio is upper bounded by $2$. Further, since the trivial lower bound is equal to $1$, we only depicted $1\leq
C/R_{sA}\leq 2$. For $n_R=9$, we observe that both the lower and upper bound are getting tighter for higher SNR. At low
SNR, the upper bound performs better than the lower bound. For $n_R=4$, $n_R=6$ and low SNR, we observe that the upper
bound is quite loose in comparison to $n_R=9$. The lower bound for $n_R=4$ is not depicted here, since it is lower than
the trivial lower bound of $1$.

\begin{figure}[htb]%htbp
     \centering
  \includegraphics[scale=0.5]{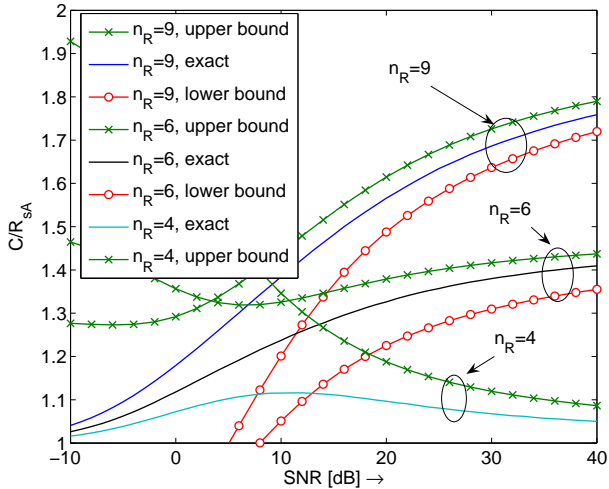}%0.5
  \caption{Ratio $C/R_{sA}$ for $n_T=8$ transmit and $n_R=4$, $n_R=6$ to $n_R=9$ receive antennas.}
   \label{fig:RatioAnaly}
\end{figure}

In Fig.~\ref{fig:AbsLoss},
\begin{figure}[htbp]
     \centering
  \includegraphics[scale=0.5]{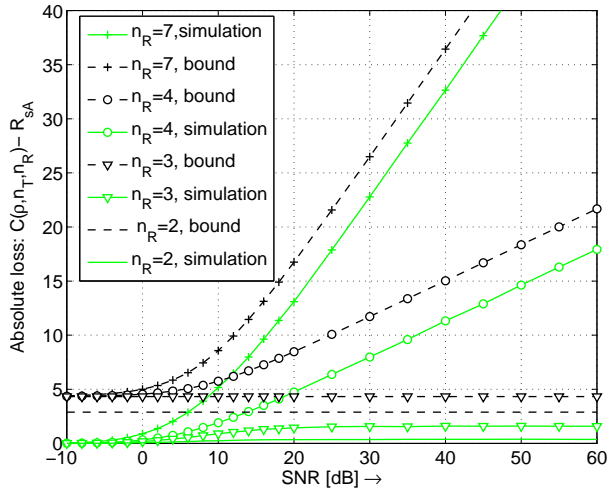}
  \caption{Absolute loss $\Delta$ for $n_T=6$ transmit and different numbers of receive antennas.} \label{fig:AbsLoss}
\end{figure}
the absolute loss $\Delta$ is depicted for $n_T=6$ transmit antennas and $n_R=2-4$ and $n_R=7$ receive antennas. From
the figure, we observe that as long as $n_T \geq 2n_R$, the slope of the absolute loss tends to a constant for high
SNR. This behavior is tracked quite well by the bound in~\eqref{eq:AbsLossAllg}, which is also depicted in the figure.

In Fig.~\ref{fig:BER_R1},
\begin{figure}[htb]
     \centering
  \includegraphics[scale=0.5]{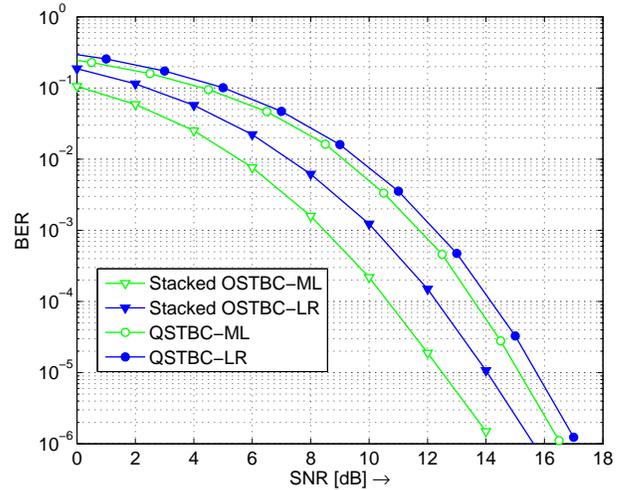}
  \caption{BER for QSTBC and the stacked OSTBC with ML and LR-ZF, 4 bit/sec/Hz.} \label{fig:BER_R1}
\end{figure}
the BER of the stacked OSTBC with QAM and the  QSTBC with 16-QAM is depicted for a transmission rate of 4 bits/sec/Hz.
Note that in order to make a fair comparison of the three transmission schemes (i.e. QSTBC, SM, and stacked OSTBC), we
analyzed a system with $n_T=n_R=4$ antennas, since for SM with suboptimal detectors it is necessary that $n_R\geq n_T$.
From the figure, we observe, that the performance of the stacked OSTBC with LR-ZF detection is comparable with the
optimal ML detection. In fact, the diversity gain of both detectors is equal and there is only a power penalty of about
$1.7$dB of LR-ZF to ML. The gap between ML and LR-ZF detection is even smaller for QSTBC. Here, the power penalty is
about $0.6$dB. Interestingly, the performance of the stacked OSTBC for both ML and LR-ZF detection is better than that
of QSTBC in the SNR region shown in the figure. However, for very high SNR and low BER, the diversity gain of $n_Tn_R$
(contrary to diversity of $2n_R$ for the stacked OSTBC) for the QSTBC will show its effect and in can be expected that
the performance of QSTBC gets better than that of the stacked OSTBC. For smaller $n_R$, this intersection point is
expected be at lower SNR values.

The bit error-rate performance of SM for BPSK and a transmission rate of 4 bits/sec/Hz is shown in
Fig.~\ref{fig:BER_RNThalf}. For comparison purposes, we also plotted the BER of the stacked QSTBC with QAM. Here, we
observe that the BER performance with ML-detection of the stacked OSTBC is better than that of SM for all SNR values.
In case of LR-ZF detection, SM performs only better than QSTBC for low SNR of about $2$dB. However, the gap in power
efficiency between ML and LR-ZF is higher for the stacked QSTBC in comparison to SM with BSPK. Note that (as
aforementioned) the small gap for SM is only due to the BPSK modulation. For higher modulation sizes, this gap is even
higher.
\begin{figure}[htb]
     \centering
  \includegraphics[scale=0.5]{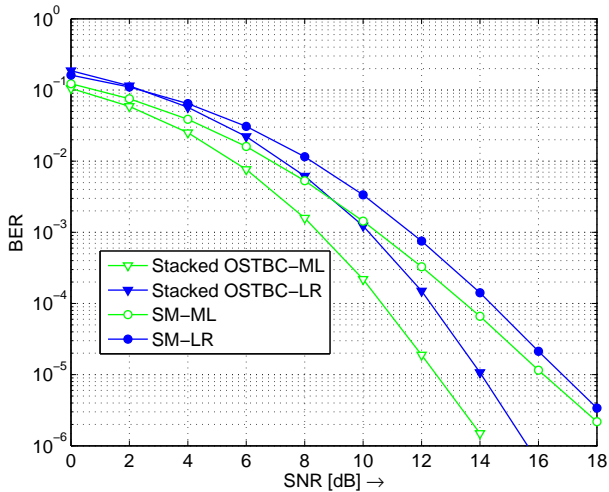}
  \caption{BER for SM and stacked OSTBC with ML and LR-ZF, 4 bit/sec/Hz.} \label{fig:BER_RNThalf}
\end{figure}
By increasing the transmission rate to $8$bit/sec/Hz, i.e. QAM for SM and 16QAM for the stacked OSTBC, we observe in
Fig.~\ref{fig:BER_RNThalf_16QAM} that the gap between ML and LR-ZF is dramatically increased in case of SM to about
$6$dB. On the other hand, the gap between ML and LR-ZF for the stacked OSTBC and 16QAM is reduced in comparison to the
gap achieved with QAM (cf. Fig.~\ref{fig:BER_RNThalf}) to about $1.3$dB. Although the performance of SM with ML
detection is better than that of the stacked OSTBC for low and moderate SNR values, for high SNR values it is the other
way around.
\begin{figure}[htb]
     \centering
  \includegraphics[scale=0.5]{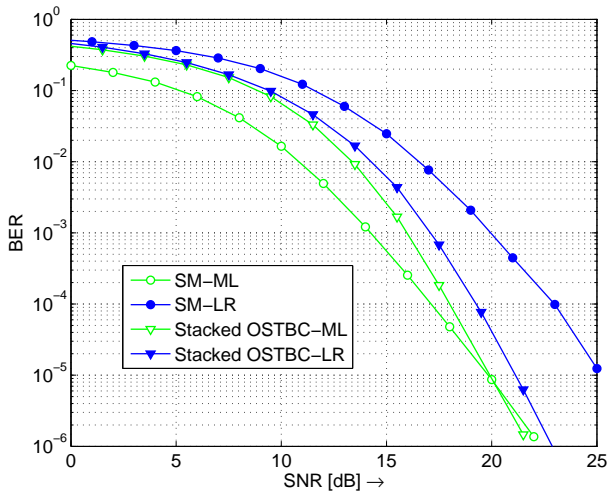}
  \caption{BER for SM and stacked OSTBC with ML and LR-ZF, 8bit/sec/Hz.} \label{fig:BER_RNThalf_16QAM}
\end{figure}
The performance of the stacked OSTBC with LR-ZF detection is better for the whole SNR range in comparison to SM, which
is of higher interest for practical applications, since the computational complexity of  the ML detector is exponential
in the transmission rate. Another disadvantage of SM is that we need at least as many receive as transmit antennas,
i.e. $n_T\leq n_R$, whereas only $\frac{n_T}{2}$ receive antennas are necessary for the stacked OSTBC. Multiple receive
antennas are only optional for the QSTBC .

\section{Conclusion}\label{sec:Conclusion}
In this paper, we analyzed the performance of stacked OSTBC in terms of the average rate. We showed, that the stacked
scheme achieves the capacity of a MIMO system in the case of $n_R=1$ receive antennas.
 Further, we showed that the MIMO capacity is at most twice the
rate achieved with the proposed scheme at any SNR. We derived lower and upper bounds for the rate achieved with this
scheme and compared it with upper and lower bounds for the capacity.

In addition to the capacity analysis, we also analyzed the error rate performance of the proposed scheme. To this end,
we combined the stacked OSTBC with a zero-forcing (ZF) detector applying lattice-reduction (LR) aided detection, since
this suboptimal detector achieves the same diversity as the optimal ML detector with only some penalty in power
efficiency. We analyzed the effect of LR on the equivalent channel generated by the stacked OSTBC, for spatial
multiplexing (SM) and QSTBC. We observed the highest gain for SM and a higher gain for the stacked OSTBC in comparison
to the QSTBC.

Finally, we illustrated the theoretical results by numerical simulations. From simulation results we observed that the
stacked scheme approaches the ergodic capacity of a MIMO system by increasing the number of transmit antennas for a
fixed number of receive antennas. Furthermore, we observed that as long as the number of transmit antennas is twice the
number of receive antennas the ratio of the capacity to the rate of the proposed scheme improves by increasing the SNR.
Regarding the simulation of the error rate performance, we observed that in the considered SNR region the stacked OSTBC
performs better in terms of BER for ML as well as for LR-aided ZF-detection than SM and QSTBC in the setup given.
Further, we observed that the gap between maximum-likelihood and LR-ZF detection is dramatically reduced in comparison
to SM schemes, especially for higher transmission rates.

\section*{Acknowledgements}
The authors thank the reviewers for their detailed and insightful comments, which significantly enhanced the quality
and readability of the paper.

\end{document}